\documentclass[twocolumn,10pt]{IEEEtran}  
\usepackage{amsmath}
\usepackage{etex} 

\usepackage{amssymb}
\usepackage{amsxtra,amsthm}
\usepackage{wasysym}% lightning                                          
\usepackage{bbold}
\usepackage{amsfonts} 
\usepackage{mathdots}
\usepackage{mathtools}
\usepackage{pdfcolmk}
\usepackage[pdftex]{graphicx}                                        % Eps Grafiken für PostScript Drucker
\usepackage{algorithm}
\usepackage{algorithmic} 
\usepackage{mathrsfs}
\usepackage{todonotes}

\usepackage[utf8]{inputenc}                                     
\usepackage[T1]{fontenc}
\usepackage{subcaption} 
\usepackage{verbatim}
\usepackage{srcltx}
\usepackage{dblfloatfix}  % Two column span figures
\usepackage{thmtools}

\usepackage[english]{babel}

\usepackage[shortlabels]{enumitem}
\usepackage[all]{xy}

\usepackage{rotating}
\usepackage{tocloft}

\usepackage[style=ieee,    
            firstinits=true, 
            url=false,
            backend=bibtex,
            doi=false,
            maxbibnames=99,
            isbn=false,
            natbib=true,
            hyperref=false,
            bibencoding=utf8]{biblatex}
\AtEveryBibitem{\clearfield{note}}   % damit mein komentare nicht geprintet werden.
\DefineBibliographyStrings{english}{%
  references = {},% replace "references" with "bibliography"  for `book`/`report`
}
\usepackage{tensor} % noetig für einige indizes. Sieht sonst scheisse aus.

\newcommand{\mgeq}{\succeq}

\newcommand{\Lone}{{{L}_1}}
\newcommand{\Ltwo}{{{L}_2}}

\let\forallalt\forall
\renewcommand{\forall}{\;\forallalt\;}

\let\refalt\ref
\renewcommand{\ref}[1]{(\refalt{#1})}

\renewcommand{\vec}{\mathbf}

\newcommand{\alp}{\ensuremath{\alpha}}
\newcommand{\bet}{\ensuremath{\beta}}
\newcommand{\Del}{\ensuremath{\Delta}}
\newcommand{\del}{\ensuremath{\delta}}

\newcommand{\Gam}{\ensuremath{\Gamma}}
\newcommand{\gam}{\ensuremath{\gamma}}

\newcommand{\lam}{\ensuremath{\lambda}}

\newcommand{\ome}{\ensuremath{\omega}}

\newcommand{\sig}{\ensuremath{\sigma}}

\newcommand{\zet}{\ensuremath{\zeta}}

\newcommand{\vlam}{\ensuremath{\boldsymbol{ \lam}}}

\newcommand{\va}{{\ensuremath{\mathbf a}}}
\newcommand{\vb}{{\ensuremath{\mathbf b}}}
\newcommand{\vc}{{\ensuremath{\mathbf c}}}

\newcommand{\ve}{{\ensuremath{\boldsymbol{\del}}}} % better and new version
\newcommand{\vn}{{\ensuremath{\mathbf n}}}                         % Vektorwertige Funktion

\newcommand{\vw}{{\ensuremath{\mathbf w}}}
\newcommand{\vx}{{\ensuremath{\mathbf x}}}
\newcommand{\vy}{{\ensuremath{\mathbf y}}}

\newcommand{\vxone}{{\ensuremath{\mathbf x}_1}}
\newcommand{\vxtwo}{{\ensuremath{\mathbf x}_2}}

\newcommand{\Rvx}{{\ensuremath{\mathbf x}^{-}}}

\newcommand{\vtx}{\ensuremath{\tilde{\mathbf x}}}

\newcommand{\uA}{{\ensuremath{\mathrm A}}}
\newcommand{\uB}{{\ensuremath{\mathrm B}}}

\newcommand{\uX}{{\ensuremath{\mathrm X}}}

\newcommand{\ua}{{\ensuremath{\mathrm a}}}
\newcommand{\ub}{{\ensuremath{\mathrm b}}}

\newcommand{\ux}{{\ensuremath{\mathrm x}}}

\newcommand{\vA}{{\ensuremath{\mathbf A}}}

\newcommand{\vD}{{\ensuremath{\mathbf D}}}

\newcommand{\vH}{\ensuremath{\mathbf H }}                         % Vektorwertige Funktion
\newcommand{\vJ}{\ensuremath{\mathbf{J}}}

\newcommand{\vM}{\ensuremath{\mathbf M}}
\newcommand{\vR}{\ensuremath{\mathbf R}}
\newcommand{\vS}{\ensuremath{\mathbf S }}                         % Vektorwertiges Ma{\ss}
\newcommand{\vT}{\ensuremath{\mathbf T}}
\newcommand{\vU}{{\ensuremath{\mathbf U}}}
\newcommand{\vV}{{\ensuremath{\mathbf V}}}
\newcommand{\vW}{{\ensuremath{\mathbf W}}}
\newcommand{\vX}{{\ensuremath{\mathbf X}}}
\newcommand{\vY}{{\ensuremath{\mathbf Y}}}

\newcommand{\valp}{\ensuremath{\boldsymbol{ \alp}}}
\newcommand{\vbet}{\ensuremath{\boldsymbol{ \bet}}}
 
\newcommand{\vzet}{\ensuremath{\boldsymbol{ \zet}}}
\newcommand{\vome}{\ensuremath{\boldsymbol{ \ome}}}
\newcommand{\zero}{{\ensuremath{\mathbf 0}}}

\newcommand{\vzero}{{\ensuremath{\mathbb 0}}}

\newcommand{\tx}{\ensuremath{\tilde{x}}}

\newcommand{\hvX}{\ensuremath{\hat{\vX}}}

\newcommand{\Galois}{\ensuremath{\mathbb{F}} }

\newcommand{\Alin}{{\ensuremath{\mathcal{A}}}}

\newcommand{\C}{{\ensuremath{\mathbb C}}}

\newcommand{\R}{{\ensuremath{\mathbb R}}}
\newcommand{\K}{{\ensuremath{\mathbb F}}} % Exclusive jetzt für Field F

\newcommand{\N}{{\ensuremath{\mathbb N}}}

\newcommand{\Fmatrix}{{\ensuremath{\mathbf F}}}

\newcommand{\id}{{\ensuremath{\mathbf I}}}

\newcommand{\RA}{\ensuremath{\Rightarrow} }

\newcommand{\LRA}{\ensuremath{\Leftrightarrow} }

\newcommand{\Pro}{\prod}
\newcommand{\Proj}{{\mathbf \Pi}}

\newcommand{\skprod}[2]{\ensuremath{ \left\langle #1,#2 \right\rangle }}
\newcommand{\sprod}[2]{\ensuremath{ \langle #1,#2 \rangle }}

\definecolor{gray}{rgb}{0.3,0.3,0.3}
{\color{black}}                      % change to gray if you want see any changes
{\color{black}}

{\color{black}}                      % change to gray if you want see any changes
{\color{black}}
\newcommand{\thmref}[1]{Theorem~\ref{#1}}     % Theorem
\newcommand{\lemref}[1]{Lemma~\ref{#1}}       % Lemma
\newcommand{\appref}[1]{Appendix~\ref{#1}}

\newcommand{\figref}[1]{Figure~\ref{#1}}

\newcommand{\noi}{\noindent}

\DeclareMathOperator{\spann}{span}

\DeclareMathOperator{\range}{range}

\DeclareMathOperator{\find}{find}

\newcommand\rank{\operatorname{rank}}

\newcommand\tr{\operatorname{tr}} 
\newcommand{\Betrag}[1]{\ensuremath{ \left|#1\right| }}

\newcommand{\Norm}[1]{\ensuremath{ \left\|#1\right\| }}
\newcommand{\norm}[1]{\ensuremath{ \|#1\| }}

\newcommand{\cc}[1]{{\ensuremath{\overline{#1}}}} % complex conjugation

\makeatletter
  \newcommand{\set}[2]{\ensuremath{%
  \setbox0=\hbox{\ensuremath{#2}}
  \dimen@\ht0
  \advance\dimen@ by \dp0
  \left\{\left.#1\rule[-\dp0]{0pt}{\dimen@}\;\right|\;#2\right\} }}
\makeatother

\newcommand{\namen}[1]{{\textsc{#1}}}           % englische w{\"o}rter als kursiv schreiben
\ifx \@paragraph \@empty
\makeatletter
\renewcommand\paragraph{\@startsection
{paragraph}{4}{\z@}{-3.5ex plus-1ex minus-.2ex}%
{1.3ex plus.2ex}{\normalfont\itshape}}
\makeatother
\fi

\renewcommand{\Re}{\ensuremath{\operatorname{Re}}}
\renewcommand{\Im}{\ensuremath{\operatorname{Im}}}
\newcommand{\vtS}{\ensuremath{\tilde{\vS}}}

\newcommand{\vtY}{\ensuremath{\tilde{\mathbf{Y}}}}

\DeclareFontFamily{U}{matha}{\hyphenchar\font45}
\DeclareFontShape{U}{matha}{m}{n}{
      <5> <6> <7> <8> <9> <10> gen * matha
      <10.95> matha10 <12> <14.4> <17.28> <20.74> <24.88> matha12
      }{}
\DeclareSymbolFont{matha}{U}{matha}{m}{n}
\DeclareFontSubstitution{U}{matha}{m}{n}

\DeclareFontFamily{U}{mathx}{\hyphenchar\font45}
\DeclareFontShape{U}{mathx}{m}{n}{
      <5> <6> <7> <8> <9> <10>
      <10.95> <12> <14.4> <17.28> <20.74> <24.88>
      mathx10
      }{}
\DeclareSymbolFont{mathx}{U}{mathx}{m}{n}
\DeclareFontSubstitution{U}{mathx}{m}{n}

\DeclareMathDelimiter{\vvvert}{0}{matha}{"7E}{mathx}{"17}
\newcommand{\sNorm}[1]{\ensuremath{ \left\vvvert#1\right\vvvert}}
\newcommand{\snorm}[1]{\ensuremath{ \vvvert#1\vvvert}}

\newcommand{\bm}{\ensuremath{\boldsymbol}}

\ifx\oast\undefined % MdSymbol defines this. IF not loaded, set it manually.
\newcommand{\oast}{\ensuremath{\circledast}}
\fi

\usepackage{filecontents}
\usetikzlibrary{positioning,decorations.pathreplacing}

  \newcommand{\tikzmark}[1]{\tikz[overlay,remember picture,baseline=(#1.base)] \node (#1) {\strut};}

  \newcommand{\tzm}[1]{\tikzmark{#1}}
  
\graphicspath{{Eps/}{pdf/}}

\newtheorem{thm}{Theorem}
\newtheorem{lemi}{Lemma}

\ifx\definition\undefined
\newtheorem{definition}{Definition}         %%%%%%%%%%%%%%%%%%%%%%%%%%%%%%
\fi
\ifx\corrolary\undefined
\newtheorem{corrolary}{Corrolary} %%%%%%%%%%%%%%%%%%%%%%%%%%%%%% 
\fi
\ifx\conjecture\undefined
\newtheorem{conjecture}{Conjecture}         %%%%%%%%%%%%%%%%%%%%%%%%%%%%%%
\fi
\ifx\thm\undefined
\newtheorem{thm}{Theorem}         %%%%%%%%%%%%%%%%%%%%%%%%%%%%%%
\fi
\ifx\lemma\undefined
\newtheorem{lemma}{Lemma}
\fi
\ifx\question\undefined
\newtheorem{question}{\color{red}{Question}}         %%%%%%%%%%%%%%%%%%%%%%%%%%%%%%
\fi
\ifx\hypothesis\undefined
\newtheorem{hypothesis}{Hypothesis}
\fi

\ifx\remark\undefined
\newenvironment{remark}{\par\vspace{1.5ex}\noindent{\em Remark\/}.}{\par\vspace{1.5ex}}
\fi

\newcommand{\unsurelist}[1]{%
  \refstepcounter{unsures}%
  \addcontentsline{uns}{unsures}%
  {\protect\numberline{\thechapter.\theunsures\ }\hspace{1em} #1}
}

\newcommand{\unsure}[1]{%
  \unsurelist{#1}\color{brown}\!#1\color{black}}

\if0
\newcommand{\wrong}[1]{%
  \unsurelist{#1}\color{red}\!#1\color{black}}

\renewcommand{\wrong}[1]{%
  \color{red}\!#1\color{black}}

\newcommand{\unsurenolist}[1]{%
  \color{brown}\!#1\color{black}}

\newcommand{\unsurenote}[2]{%
  \unsurelist{#1}\color{brown}\!#1\footnote{\color{brown}#2\color{black}}\color{black}}
\fi

\newcommand{\wrong}[1]{}

\newcommand{\unsurenolist}[1]{}

\newcommand{\unsurenote}[2]{}

\ifx\unsure\undefined
 \newcommand{\unsure}[1]{}
\else
 \renewcommand{\unsure}[1]{} 
\fi
\ifx\wrong\undefined
  \newcommand{\wrong}[1]{}
\else
  \renewcommand{\wrong}[1]{}
\fi
  
\ifx\unsurenolist\undefined
  \newcommand{\unsurenolist}[1]{}
\else
  \renewcommand{\unsurenolist}[1]{}
\fi

\ifx\unsurenote\undefined
  \newcommand{\unsurenote}[2]{}
\else
  \renewcommand{\unsurenote}[2]{}
\fi

\renewcommand{\alp}{\ensuremath{\alpha}}
\renewcommand{\Lone}{\ensuremath{{N_1}}}
\renewcommand{\Ltwo}{\ensuremath{{N_2}}}

\newcommand{\Ni}{\ensuremath{{N_i}}}
\newcommand{\Nj}{\ensuremath{{N_j}}}
\newcommand{\Nij}{\ensuremath{{N_{ij}}}}

\newcommand{\vxgt}{\ensuremath{\vx_o}} % ground truth x
\newcommand*\xbar[1]{%
   \hbox{%
     \vbox{%
       \hrule height 0.5pt % The actual bar
       \kern-0.1ex%         % Distance between bar and symbol
       \hbox{%
         \kern-0.0em%      % Shortening on the left side
         \ensuremath{ #1}%
         \kern-0.1em%      % Shortening on the right side
       }%
     }%
   }%
} 
\newcommand*\sxbar[1]{%
   \hbox{%
     \vbox{%
       \hrule height 0.5pt % The actual bar
       \kern-0.3ex%         % Distance between bar and symbol
       \hbox{%
         \kern-0.0em%      % Shortening on the left side
         \ensuremath{\scriptstyle #1}%
         \kern-0.1em%      % Shortening on the right side
       }%
     }%
   }%
} 
\newcommand{\vxct}{\xbar{\vx_{\phantom{1}} \rule{-1.2ex}{1.6ex}^{_{\scriptstyle -}}}}

\newcommand{\vxonect}{\xbar{\vx_1 \rule{-1.2ex}{1.6ex}^{_{\scriptstyle -}}}}

\newcommand{\vxonet}{{\vx_1 \rule{-1.2ex}{1.6ex}^{_{\scriptscriptstyle -}}}}
\newcommand{\vxit}{{\vx_i \rule{-1.2ex}{1.6ex}^{_{\scriptscriptstyle -}}}}

\newcommand{\vxtwot}{{\vx_2 \rule{-1.2ex}{1.6ex}^{_{\scriptscriptstyle -}}}}
\newcommand{\vaict}{\xbar{\va_i \rule{-1.2ex}{1.7ex}^{_{\scriptstyle -}}}}

\newcommand{\vxjct}{\xbar{\vx_j \rule{-1.2ex}{1.6ex}^{_{\scriptstyle -}}}}
\newcommand{\vxict}{\xbar{\vx_i \rule{-0.9ex}{1.6ex}^{_{\scriptstyle -}}}}

\newcommand{\vxtwoct}{\xbar{\vx_2 \rule{-1.2ex}{1.6ex}^{_{\scriptstyle -}}}}

\newcommand{\svact}{\sxbar{\va\rule{-1.2ex}{1.4ex}^{_{\ -}}}}
\newcommand{\svbct}{\sxbar{\vb \rule{-1.2ex}{1.4ex}^{_{\ -}}}}

\newcommand{\svaonetwoct}{\sxbar{\va_{1,2} \rule{-0.9ex}{1.3ex}^{_-}}}
\newcommand{\svxict}{\sxbar{\vx_i \rule{-0.9ex}{1.3ex}^{_-}}}

\newcommand{\svxtwoct}{\xbar{\vx_2 \rule{-0.9ex}{1.3ex}^{_-}}}

\newcommand{\Nip}{\ensuremath{N_{i'}}}
\newcommand{\Njp}{\ensuremath{N_{j'}}}
\newcommand{\vWminus}{\ensuremath{\vW_{\!-}}}

\usepackage{multirow}
\def\minus{%
  \setbox0=\hbox{-}%
  \vcenter{%
    \hrule width\wd0 height \the\fontdimen8\textfont3%
  }%
}

\renewcommand{\ve}{\ensuremath{{\mathbf e}}}
\newcommand{\uareci}{\ensuremath{\ua^-}}
\newcommand{\ubreci}{\ensuremath{\ub^-}}

\newcommand{\Wperp}{\ensuremath{W_{\!\perp}}}
\newcommand{\vwperp}{\ensuremath{\vw}}
\newcommand{\vwperpi}{\ensuremath{\vw_i}}
\newcommand{\vwperpone}{\ensuremath{\vw_{1}}}
\newcommand{\vwperptwo}{\ensuremath{\vw_{2}}}

\renewcommand{\uareci}{\ensuremath{\ua}} % simplify notation for journal version
\renewcommand{\ubreci}{\ensuremath{\ub}}

\begin{document}
  \title{Stable Blind Deconvolution over the Reals from Additional Autocorrelations}

  \author{
  \IEEEauthorblockN{Philipp Walk\IEEEauthorrefmark{1}  and Babak Hassibi\IEEEauthorrefmark{1}\\}
  \IEEEauthorblockA{\IEEEauthorrefmark{1}Department of Electrical Engineering, Caltech, Pasadena, CA 91125\\
   Email: \{pwalk,hassibi\}@caltech.edu }
  }

  \maketitle

  \begin{abstract}
     Recently the one-dimensional time-discrete blind deconvolution problem was shown to be solvable uniquely, up to a
     global phase, by a semi-definite program for almost any signal, provided its autocorrelation is known.  We will
     show in this work that under a sufficient zero separation of the corresponding signal in the $z-$domain, a stable
     reconstruction against additive noise is possible. Moreover, the stability constant depends on the signal dimension
     and on the signals magnitude of the first and last coefficients. We give an analytical expression for this
     constant by using spectral bounds of Vandermonde matrices.
  \end{abstract}

  \section{Introduction}
  % Introduction
%

One-dimensional blind deconvolution problems occur in many signal processing applications, as in imaging and digital
communication over wired or wireless channels, where the blurring kernel and the channel are modeled as \emph{linear
time invariant} (LTI) systems, which have to be blindly identified or estimated. Since the convolution is a product of
two input signals in the frequency domain, it can be generated by uncountably many input signals. Therefore, the
deconvolution problem is always ill-posed and hence additional constraints are necessary to eliminate such ambiguities,
which to find are already a challenging task see for example \cite{WJPH17,CM14c}.  If the receiver has some additional
statistical knowledge of the LTI system, as given for example by second or higher order statistics, blind channel
equalization and estimation techniques were already developed in the $90'$s; see for example in
\cite{TXK91,DKAJ91,TXHK95}.  In \cite{XLTK95} a new blind channel identification for multi-channel \emph{finite impulse
response} (FIR) systems was proposed, which refers to a \emph{single input multi-output} (SIMO) system where an $N_1$
dimensional input signal is sent via $M$ FIR channels of dimension $N_2$. The authors showed in Thm.2 that the property
of no common zeros of the $M$ channels is a necessary condition for unique identification.

If no statistical knowledge of the data and the channel is available, for example, for fast fading channels, one can
still ask under what conditions on the data and the channel is blind identification possible. Necessary and sufficient
conditions in a multi-channel setup where first derived in \cite{XLTK95,GN95} and continuously further developed, for a
nice summary see \cite{A-MQH97}. The big disadvantage of all these techniques lies in lack of efficiency, since the
algorithms for identification are iterative. 
In the recent works \cite{JH16,WJH17b} the authors formulated a semi-definite program (SDP) which solves, for almost all
input signals, the deconvolution problem up to a global phase if additionally the signals in the $z-$domain are coprime
polynomials and their autocorrelations are both known. The co-primeness of the signals was already shown in
\cite{XLTK95} to be a necessary condition for a unique blind deconvolution. 
To obtain a convex optimization the convolution and autocorrelations were lifted to a matrix recovery problem. This idea
of combining in one linear measurement the (cross) correlation and autocorrelations was first used in a phase retrieval
problem known as vectorial phase retrieval \cite{RDN13}.

In this work we prove  stability of the SDP against additive noise in the real-case.  The stability constant depends
hereby only on the zero separation of the input signals $z-$transforms and on their first and last coefficients
magnitude, which defines non-linear constraints in the time domain.  Such a stability analysis for pure deterministic
input signals is of tremendous interest for any practical application of deconvolution.  Moreover, as soon as one of the
inputs is randomly chosen such a zero separation holds with high probability, which therefore demonstrate the success of
many randomized  deconvolution results \cite{ARR12,LLSW16,JKS17,GKK17,GKK15}.  

\subsection{Notation}

We will use for $\Galois$ either the real $\R$ or complex $\C$ field. Although, we will formulate most of the results in the
complex case we derive the stability constant only for the real case due to space of interest and postpone the
complex case to an upcoming  work.  For an integer $N$, we will denote the first $N$ non-negative integers by
$[N]:=\{0,1,\dots,N-1\}$.  Small letters denote scalars in $\K$, bold small letters vectors in $\K^N$  and bold capital
letters denote $N\times M$ matrices in $\K^{N\times M}$. By $\id_N$ we denote the $N\times N$ identity matrix and by
$\vzero_N$ the all zero matrix. We denote by $\cc{x}=\Re x -i\Im x$ the complex-conjugate of $x\in\C$ and by
$(\cdot)^*=\cc{(\cdot)^T}$ the complex-transpose of a vector or matrix. For a vector $\vx\in\K^N$ we will denote its
time-reversal (reflection) by $\vx^-$ given component-wise for $k\in[N]$ by $x_k=x_{N-(k-1)}$.  The linear spaces are
equipped with the scalar product, given by
\begin{align} \begin{split} \skprod{\vx}{\vy}&=\vy^*\vx = \sum_{k} x_k \cc{y_k} \quad\text{for}\quad \vx,\vy\in\K^N\\
    \skprod{\vX}{\vY}&=\tr({\vY}^*\vX)=\sum_{k}(\vY^*\vX)_{k,k}  \quad\text{for}\quad \vx,\vy\in\K^N.
  \end{split}
\end{align}
We denote by $\lam_k(\vA)$ the $k$th eigenvalue of a Hermitian matrix $\vA\in\K^{N\times N}$ which we will order by
increasing values, i.e. $\lam_1(\vA)\leq \lam_2(\vA)\leq\dots\leq \lam_N(\vA)$. Moreover, for an arbitrary matrix
$\vA\in \K^{N\times M}$ we will denote the singular values by $\sig_k(\vA)=\sqrt{\lam_k(\vA\vA^{*})}\geq 0$ for
$k\in\{1,\dots,\min\{N,M\}\}$, see \cite{HJ90}.  We will use the $\ell^p-$norms for $p\in[1,\infty)$ for $\va\in\K^N$
  and $\vA\in\K^{N\times M}$
\begin{align} 
  \Norm{\va}_p\!=\!\Big(\sum_k |a_k|^p\Big)^{1/p} \quad\!,\quad \sNorm{\vA}_p \!=\!\Big(\sum_k
  \sig_k^p(\vA)\Big)^{1/p}, 
\end{align}
We use for the Hilbert norms $\Norm{\va}_2\!=\!\Norm{\va}$ resp. $\sNorm{\vA}_2\!=\!\sNorm{\vA}$ and define for
$p=\infty$
\begin{align} 
  \Norm{\va}_{\infty}=\max_k |a_k|\ ,\quad \sNorm{\vA}_{\infty}= \max_k
  \sig_k=\sig_{\max\{N,M\}}\label{eq:inftynorm}.  
\end{align}

  \section{Blind Deconvolution from Additional Autocorrelations via SDP}
  % Blind Deconvolution form additional autocorrelations: An SDP
%

In this work we will only consider one-dimensional convolution, i.e., the convolution between two complex-valued vectors
$\vx_1\in\C^{\Lone}$ and $\vx_2\in\C^{\Ltwo}$, given component-wise by 
\begin{equation} 
  y_k=(\vec{x}_1\ast \vec{x}_2)_k=\sum_{l}x_{1,l} x_{2,k-l}\label{eq:convtime} 
\end{equation}
for $k\in[\Lone+\Ltwo]$.
Let us denote by $\vxct$ the \emph{conjugate-time-reversal} or \emph{reciprocal} of $\vx\in\C^N$ given by
$x_k=\cc{x_{N-1-k}}$ for $k\in [N]$. Then the \emph{correlation} between $\vx_1$ and $\vx_2$ is given by
\begin{align} 
  \va_{1,2} = \vx_1* \vxtwoct\label{eq:correlation}. 
\end{align}
If $\vx_2=\vx_1$ then we write $\va_{1,1}=\va_1=\vx_1*\vxonect$ which is called the \emph{autocorrelation} of $\vx_1$.  The
autocorrelation in the Fourier domain is given by the absolute-squares of the Fourier transform of $\vx_1$, i.e., the
phase information of the signal in the Fourier domain is missing. The recovery of the signal from its absolute-square
Fourier measurements is known as the \emph{phase retrieval problem} and therefore a special case of the deconvolution
problem. However, the convolution and even autocorrelation obtains uncountably many ambiguities see \cite{WJPH17,BP15}.
To resolve the ambiguities and to formulate a well-posed deconvolution or phase retrieval problem we need additional
constraints or measurements. For example, if we can measure the auto and the cross correlation separately, we can
resolve all ambiguities up to a global phase. Moreover, the reconstruction can be performed by an SDP by lifting  the
measurements to the matrix domain, i.e., we express the measurements as linear mappings on positive-semidefinite
rank$-1$ matrices \cite{JH16}. Such lifting methods to relax to a convex problem can be used for phase retrieval
\cite{CESV13} and arbitrary bi- or multi-linear measurements \cite{CR09},\cite{Gro11}.
For this we stack both vectors $\vx_1$ and $\vx_2$ together to obtain the vector $\vx=[\vx_1, \vx_2]\in\C^{N}$ in
$N:=\Lone+\Ltwo$ dimensions, which, if lifted to the matrix domain reveals a $2\times2$ block matrix structure   
\begin{align} 
  \vx\vx^*= \begin{pmatrix}\vx_1\\ \vx_2\end{pmatrix}\begin{pmatrix}\vx_1^* &
    \vx_2^*\end{pmatrix}=\begin{pmatrix} \vx_1\vx_1^* & \vx_1\vx_2^*\\ \vx_2\vx_1^* &
    \vx_2\vx_2^*\end{pmatrix}\label{eq:xx4block}.  
\end{align}
To define the linear measurement map $\Alin$  we have to introduce the $N\times N$  \emph{shift} or
\emph{elementary Toeplitz matrix}  and the $\Nij\times \Ni$ \emph{embedding matrix} with $\Nij=\Ni+\Nj$ for any $i,j\in\{1,2\}$ by
\begin{align} 
  \vT_N\!=\!
  \setlength\arraycolsep{2pt}
  \begin{pmatrix} 0 & \!\cdots\! \!& 0 & 0 \\ 
                  1 & \!\cdots\! \!& 0 & 0\\[-0.3em] 
                  \vdots &\!\diagdown\!\! &  \vdots & \vdots\\ 
                  0 & \!\cdots\!\! & 1 & 0\\ \end{pmatrix},\ \
                \Proj_{i+j,i}\!=\!\Proj_{\Nij,\Ni}&\!=\!\begin{pmatrix} \id_{\Ni,\Ni}\\ \vzero_{\Nj, \Ni}\end{pmatrix}.
                  \label{eq:elementtoeplitz}
\end{align}
Then, the $\Ni\times \Nj$ rectangular shift matrices are defined as, 
\begin{align} 
  \vT^{(k)}_{i,j}=\vT^{(k)}_{\Ni,\Nj}=\Proj^T_{i+j,i}\vT^{\Nj-1-k}_{\Nij}\Proj_{i+j,j}\label{eq:LiLj} 
\end{align}
for $k\in[\Nij-1]$, where we set $\vT^l_N:=(\vT^{-l}_N)^T$ if $l<0$. Then, the correlation
\eqref{eq:correlation} between the vectors $\vx_i$ and $\vx_j$ are given component-wise as 
\begin{align} 
  a_{i,j,k}=(\vx_i*\vxjct)_{k} = \vx_j^* \vT^{(k)}_{j,i}\vx_i =
  \tr(\vT^{(k)}_{j,i}\vx_i\vx_j^*). \label{eq:aijk} 
\end{align}
Embedding the translation matrices into $N$ dimensions
defines the linear map $\Alin\colon \C^{N\times N}\to \C^M$ component wise for $k\in[\Ni+\Nj-1]$ and $i,j\in\{1,2\}$ by
\begin{equation}
  \Alin_{i,j,k}(\vX)\!=\!\skprod{\vX}{\vA_{i,j,k}}\!=\!\tr(\vA_{i,j,k}\vX^*) \quad\text{with }  \label{eq:Aijtrace}\\
\end{equation}
\begin{equation}
  \begin{split}
   \vA_{1,1,k} \!&=\! 
   \setlength\arraycolsep{3pt}
        \begin{pmatrix} 
          \!\vT^{(\!k\!)}_{1} &\!\!\vzero\\ 
          \vzero  &\!\!\vzero
        \end{pmatrix},
        \vA_{2,2,k} \!=\! \begin{pmatrix} \vzero & \!\!\vzero\\ 
         \vzero & \!\!\vT^{(\!k\!)}_{2}
       \end{pmatrix},\\
       \vA_{1,2,k}\! &=\! 
       \begin{pmatrix} 
         \vzero &  \!\!\vzero\\ 
         \! \vT^{(k)}_{2,1}& \!\! \vzero\end{pmatrix} =\vA_{2,1,N-1-k}^T ,\end{split}\label{eq:4maskmeasurements}
\end{equation}
where we used the notation $\vT_i=\vT_{i,i}$.  By stacking the linear maps \eqref{eq:4maskmeasurements} in
lexicographical order together  we get the measurement map $\Alin=(\Alin_{1,1},\Alin_{1,2},\Alin_{2,1},\Alin_{2,2})$
defining the $M=4N-4$ complex-valued measurements 
\begin{align} 
  \Alin(\vx\vx^*\!)
    &\!=\! \begin{pmatrix}\Alin_{1,1}(\vx\vx^*\!)\\ \Alin_{1,2}(\vx\vx^*)\\\Alin_{2,1}(\vx\vx^*) \\\Alin_{2,2}(\vx\vx^*) \end{pmatrix} 
     \!= \!\begin{pmatrix} \vxone*\vxonect \\ \vxone*\vxtwoct\\ \vxtwo*\vxonect\\ \vxtwo*\vxtwoct   \end{pmatrix}
     \!=\!\begin{pmatrix}\va_{1,1}\\\va_{1,2}\\\va_{2,1}\\\va_{2,2}\end{pmatrix}\!=\!\vb\label{eq:3Nb}.
\end{align}
Note, since the auto-correlations $\va_{i}$ are  conjugate-symmetric, i.e., $\vaict=\va_i$, we
would require only $2N-5$ complex-valued linear measurements of $\vx\vx^*$. 
Hence, the \emph{blind deconvolution problem}\footnote{If we set $\vx=[\vx_1,\svxtwoct]$ the cross correlation becomes
a cross convolution, i.e., $\va_{1,2}=\vy$ whereas the auto-correlations are not changing due to the commutativity of
the convolution.} from $\va_{1,2}$ and the additional autocorrelations
$\va_1$ and $\va_2$ is re-cast as a \emph{generalized phase retrieval problem} from 
$4N-3$ Fourier magnitude measurements given by
\begin{align} 
  \left|\Fmatrix_{2N\!-\!1} \begin{pmatrix}\vx \\ \zero\end{pmatrix}\right|^2\!, \ 
  \left|\Fmatrix_{2\Lone\!-\!1}\begin{pmatrix}\vxone\! \\ \zero\end{pmatrix}\right|^2\!\text{ and }
  \left|\Fmatrix_{2\Ltwo\!-\!1} \begin{pmatrix}\vxtwo\! \\ \zero\end{pmatrix}\right|^2\!\label{eq:gprmeas}
\end{align}
where $\Fmatrix_N$ denotes the $N\times N$ unitary Fourier matrix. In fact, the Fourier measurements on $\vx_1$ and
$\vx_2$ can be obtained by masked Fourier measurements of $\vx$, see \cite{Jag16}. To construct an explicit measurement
ensemble of the smallest size, which was conjectured to be $4N-4$ in \cite{BCMN13}, is a challenging and still open task
for the generalized phase retrieval problem, see for example \cite{Kec16,CEHV14}. Note, that we demand in
\eqref{eq:gprmeas} also a co-prime structure on the separated parts $\vx_1$ and $\vx_2$. To see that our correlation
measurements give indeed access to the autocorrelation of $\vx$ we split it in its single parts 
\begin{align} 
  \va&=\vx*\cc{\Rvx}
     = \begin{pmatrix}\vxone\\\vxtwo\end{pmatrix}*\begin{pmatrix}\vxtwoct\\ \vxonect\end{pmatrix}\notag\\
     &= \!\begin{pmatrix}\!\vxone\!\!\\\zero \end{pmatrix}\!*\!\begin{pmatrix}\zero\\
        \!\vxonect\!\end{pmatrix} 
      \!+\!\begin{pmatrix}\!\vxone\!\!\\\zero \end{pmatrix}\!*\!\begin{pmatrix}\!\vxtwoct\!\\
        \zero\end{pmatrix} \!+\!\begin{pmatrix}\zero\\ \!\vxtwo\!\! \end{pmatrix}\!*\!\begin{pmatrix}\!\vxtwoct\! 
        \\\zero\end{pmatrix}
  \!+\!\begin{pmatrix}\zero\\ \!\vxtwo\!\! \end{pmatrix}\!*\!\begin{pmatrix}\zero\\ \!\vxonect\!\end{pmatrix}\notag\\
  &=  \begin{pmatrix}\zero_{\Ltwo} \\ \va_{1,1} \\ \zero_{\Ltwo} \end{pmatrix} 
  +\begin{pmatrix} \va_{1,2}\\  \zero_{\Lone}\\ \zero_{\Ltwo}  \end{pmatrix} 
  +\begin{pmatrix} \zero_{\Lone}\\ \va_{2,2}\\ \zero_{\Lone}\end{pmatrix} 
  +\begin{pmatrix} \zero_{\Lone} \\ \zero_{\Ltwo}\\ \va_{2,1}  \end{pmatrix} 
  \label{eq:difautocor}.  
\end{align}
As can be seen, for any choice of $\Lone,\Ltwo$, the cross-correlations $\va_{1,2}$ and $\va_{2,1}$ are always separated
in time by one time slot\footnote{Exactly this ``$0$`` measurement, due to symmetry of the correlations, cost us the $1$
extra measurement in our phase retrieval setting.}, and can therefore exactly be obtained by subtracting the known
auto-correlations $\va_1$ and $\va_2$. Hence, the unitary Fourier transforms of the measurements in \eqref{eq:gprmeas}
are $\va,\va_1$ and $\va_2$ which are equivalent to $\va_{1,2},\va_1$ and $\va_2$. Therefore, the  following
\thmref{thm:4correlation} can be seen as a unique reconstruction of $\vx$, up to a global phase, via an SDP from $4N-3$
masked Fourier magnitude measurements by only assuming that the $z-$transforms of $\vxone$ and $\vxtwo$ are co-prime.
Here, the $z-$transform of $\vx\in\C^{N+1}$ is given by
\begin{align}
  \uX(z):=\sum_{k=0}^{N} x_k z^{-k}= x_0z^{-N}\Pro_{k=1}^N(z-\zeta_k)\label{eq:ZformX}
\end{align}
and defines a polynomial in $z^{-1}$ of order $N$ if $x_{0}\not=0$. If $\zeta_{1,k}$ and $\zeta_{2,k}$ are the zeros
of $\uX_1$ and $\uX_2$, then we call them co-prime if $\zeta_{1,k}\not=\zeta_{2,l}$ for all $l,k$ (no common factor). \\
Let us denote by $\C^{L}_{0,0}:=\set{\vx\in\C^L}{x_0\not=0\not=x_{L}}$.  Then, \cite[Thm.III.1]{JH16} and extended to
the purely deterministic case \cite{WJPH17}, it holds the following theorem. 
\renewcommand{\vxgt}{\vx}
%
%%%%%%%%%%%%%%%%%%%%%% THEOREM %%%%%%%%%%%%%%%%%%%%%%%%%%%%%%%%%%%%%%%%%%%%%%%%%%%%%%%%%%%%%%%%%%%%%%%%%%%%%%%%%%%
\begin{thm}\label{thm:4correlation}
  Let $\vx_1\in\C_{0,0}^{\Lone}$ and $\vx_2\in\C_{0,0}^{\Ltwo}$ such that their $z-$transforms  are co-prime. Then
  $\vxgt=[\vx_1,\vx_2]\in\C^N$ with $N=\Lone+\Ltwo$ can be recovered uniquely up to global phase from the $4N-4$
  measurement $\vb$ defined  in \eqref{eq:3Nb} by solving the feasible convex program
  \begin{align} 
    \find \vX\in\C^{N\times N}\quad\text{s.t.}\quad \begin{split}\Alin(\vX)=\vb\\ \vX\mgeq 0\end{split}
    \label{eq:fcp3N} 
  \end{align}
  which has $\hat{\vX}=\vxgt\vxgt^*$ as the unique solution.  
\end{thm}
%%%%%%%%%%%%%%%%%%%%%%%%%%%%%%%%%%%%%%%%%%%%%%%%%%%%%%%%%%%%%%%%%%%%%%%%%%%%%%%%%%%
%
\begin{remark} 
  From a singular value decomposition of $\hvX$ one can identify up to a global phase $\vx$ as the eigenvector
  corresponding to the largest eigenvalue. By knowing the
  dimensions $\Lone$ and $\Ltwo$ one can also identify $\vx_1$ and $\vx_2$  up to a common global phase.  
\end{remark}

  \section{Stability Analysis of the SDP}
  % Stability analysis

The feasible SDP in \thmref{thm:4correlation} is in the noise-less case equivalent to the convex problem
\begin{align}
  \min_{\vX\mgeq 0} \Norm{\Alin(\vX)-\vb}
\end{align}
for any $\vb:=\Alin(\vx\vx^*)$.
If the observations are disturbed by additive noise $\vn\in\C^{4N-4}$, such that
\begin{align}
  \vb=\Alin(\vx\vx^*)+\vn.
\end{align}
Then the \emph{denoised SDP} (least-square minimization over a convex cone)
\begin{align}
  \min_{\vX\mgeq 0} \Norm{\Alin(\vX)-\vb}\label{eq:SDPnoisy}
\end{align}
is robust against noise if the solution $\hat{\vX}$ obeys
\begin{align}
  \vvvert\hat{\vX}-\vx\vx^*\vvvert\leq C\Norm{\vn}\label{eq:noisySDP}
\end{align}
where $C>0$ is a constant independent of the chosen vectors $\vx_1,\vx_2$ and only depends on the dimensions, see
\cite[Thm.2.2]{CSV12} and \cite[Thm3]{DH12}.
For a general bilinear problem, such a stability constant $C$ would at least depend on a rank$-2$ null-space property or
\emph{restricted isometry property} (RIP) of the linear map $\Alin$, see \cite{RFP10,JW13}. 
However, to apply our proof technique in the noisy case with the construction of an (inexact) dual certificate, we only
need the RIP to hold locally, i.e., around the ground truth $\vxgt\vxgt^*$. For the quadratic case see \cite{CSV12,DH12}
and for the non-quadratic case \cite{ARR12},\cite[Condition 5.1]{LLSW16} and \cite[Def.4.1]{JKS17}.

\subsection{Local Stability on the Tangent Space}

The crucial part for the proof of the injectivity and for the stability is the explicit construction of a \emph{dual
certificate} in \cite{JH16,WJPH17}. The dual certificate is given by a linear combination of the measurement matrices
\eqref{eq:4maskmeasurements}  depending on the observed measurements \eqref{eq:3Nb} usually in a complex algebraic
manner. Fortunately, for the correlation type map $\Alin$ the measurements are linked to the measurement matrices by
\emph{banded Toeplitz matrices} generated by unknown vectors $\vx_1$ and $\vx_2$. Morepreciselcy, the linear convolution
between $\vx_i\in\C^{\Ni},\vx_j\in\C^{\Nj}$  for any $i,j\in\{1,2\}$ is given by applying the $\Nij\times N_j$ banded
Toeplitz matrix ($\Nij=\Ni+\Nj$), where for the correlation we need also the $\Nj\times\Nj$ time-reversal matrix
\begin{align}
  \vT_{j,\vx_i} \!\!=\!\sum_{k=0}^{\Ni-1} \!\!x_{i,k} \vT^k_{\Nij\!-\!1} \Proj_{\Nij\!-\!1,\Nj},
\quad
\setlength\arraycolsep{2.2pt}
\vR_{\Nj}\!=\!\begin{pmatrix} 
    0 &\!\!\cdots \!\!  & 1\\[-0.5em]
  \vdots &  \!\diagup\!\!  & \vdots \\
   1 &   \!\!\cdots \! \!&0 
  \end{pmatrix}\label{eq:Tform}
\end{align}
where $\vT_N$ denotes the elementary Toeplitz matrix \eqref{eq:elementtoeplitz}.
Hence the matrix form of the convolution and correlation is given by
\begin{align}
  \vx_i*\vx_j&=\vT_{j,\vx_i}\vx_j\label{eq:Txij}\\
  \vx_i*\vxjct&=\vT_{j,\vx_i}\vR_{\Nj}\cc{\vx_j}.\label{eq:Hxij}
\end{align}
Concatenating two such Toeplitz matrices defines the $N\times N$ \emph{Sylvester matrix}
\begin{align}
  \vS_{\vx_2^0,-\vx_1^0}
  \setlength\arraycolsep{2pt}
  \!=\!\begin{pmatrix}\vT_{1,\vx_2^0} & \vT_{2,-\vx_1^0}\end{pmatrix}
  \setlength\arraycolsep{2pt}
  \!=\!\begin{pmatrix}\vT_{1,\vx_2} &
    \vT_{2,-\vx_1}\\ \zero^T & \zero^T\end{pmatrix}\label{eq:sylvesterx}
\end{align}
where we used the notation 
\begin{align}
  \vx_i^0=\begin{pmatrix}\vx_i\\ 0\end{pmatrix}\in\C^{\Ni+1}\label{eq:xaddzero}.
\end{align}
This allows us to represent any convolution difference
for $\vy_1\in\C^{\Lone}$ and $\vy_2\in\C^{\Ltwo}$ as a matrix equation 
\begin{align}
  \vS_{\vx}\vy=  \vS_{\vx_2^0,-\vx_1^0} \begin{pmatrix}\vy_1\\ \vy_2\end{pmatrix} 
  &= \begin{pmatrix} \vx_2*\vy_1 -\vx_1*\vy_2 \\ 0 \end{pmatrix} \label{eq:convdiff}.
\end{align}
We will show in \appref{sec:dualcertificate}, that the dual certificate for each ground truth signals
$\vx_1\in\C^{\Lone}$ and $\vx_2\in\C^{\Ltwo}$ is given by
\begin{align}
  \vW=\vS_{\vx_2^0,-\vx_1^0}^*\vS_{\vx_2^{0},-\vx_1^0}\label{eq:dualcertif}.
\end{align}    

The injectivity of the convolution-type map $\Alin$ and the stability is then fully determined by the singular value
properties of the Sylvester matrix, see \appref{app:sylvester}.  Therefore we have to extend the dual certificate in
\cite[Lem.3]{WJPH17} (injectivity) to a local stability of $\Alin$ on the tangent space
$T_\vx:=\set{\vx\vy^*+\vy\vx^*}{\vy\in\C^N}$, i.e., for each ground truth signal $\vx\in\C^N$ we have to show that there
exists some $\gam=\gam(\vx_1,\vx_2)>0$ such that it holds 
\begin{align}
  \Norm{\Alin(\vY)} \geq \gam \sNorm{\vY}\quad,\quad \vY\in T_{\vx}   \label{eq:localRIPAlin},
\end{align}
see also \cite{CSV12,DH12}.

\begin{remark}
  Note, that each $\Lone,\Ltwo$ specify a different linear map $\Alin$ even if $\Lone+\Ltwo=N=\text{const.}$. Therefore,
  $\gam$ depends on $\Lone$ and $\Ltwo$  and not only on their sum.  Since $T_{\vx}$ is a linear space it holds
  $\vY-\vtY\in T_{\vx}$ for all $\vY,\vtY\in T_{\vx}$. 
  Moreover, $T_{\vx}$ does not include all rank$-1$ differences, i.e., not all rank$-2$ matrices and
  \eqref{eq:localRIPAlin} therefore only obeys a \emph{local stability of $\Alin$ for rank$-2$ matrices} around the
  ground truth $\vx\vx^*$ or a \emph{local rank$-2$ RIP}, which is much less strict condition than a rank$-2$ RIP.  In
  fact, a rank$-2$ RIP for convolutions can not hold since the difference of arbitrary convolutions can vanish. Even to
  control the norm of the convolution can be a challenging task for certain signals \cite{WJP15}.
\end{remark}
Since the norms are absolutely homogeneous and  $\Alin, T_{\vx}$ are linear, we need to show \eqref{eq:localRIPAlin}
only for $\vtx=\vx/\Norm{\vx}\in\K^N$ 
\begin{align}
  T_{\vtx}:=\set{\vtx\vy^*+\vy\vtx^*}{\vy\in\K^N}
\end{align}
for $\K\in\{\C,\R\}$.  Let us assume from here that $\vx=\vtx$. In fact, the tangent space $\vT_{\vx}$ of rank two
matrices refers to a sum of two convolutions which is parameterized by two unknown as $\vy=[\vy_1,\vy_2]$ with same
dimensions as $\vx_1$ and $\vx_2$. The structure of convolution sums is given by Sylvester matrices. Hence, to exploit
their structure  we need to parameterize the matrices in $T_{\vx}$ by vectors in $\K^N$. Furthermore, we can easily
represent the Schatten $2-$norm of $\vY$ in $\vy$ as
\begin{align}
  \begin{split}
  \sNorm{\vY}^2&=\tr((\vx\vy^*+ \vy\vx^*)^* (\vx\vy^*+\vy\vx^*))\\ 
                & =2\tr(\vx\vy^*\vy\vx^*)+ \tr(\vy\vx^*\vy\vx^*)+\tr(\vx\vy^*\vx\vy^*)\\
                &=2 \Norm{\vy}^2 + 2\Re\{\skprod{\vx}{\vy}^2\}.
  \end{split}
  \label{eq:Ynorm}
\end{align}
Then \eqref{eq:localRIPAlin} is equivalent to
\begin{align}
  \Norm{\Alin(\vx\vy^*\!+\!\vy\vx^*)}^2\!\geq\! 2\gam^2 (\Norm{\vy}^2\!+\!\Re\{\skprod{\vx}{\vy}^2\})
  \label{eq:localRIPy}
\end{align}
for all $\vy\in\K^N$. Unfortunately, for $\K=\C$ we face two problems: First the left hand side of \eqref{eq:localRIPy}
can not be written as a quadratic form in $\vy$, due to the alternate complex-conjugation,  and second the right hand
side vanishes for some $\rho\geq 0$ if $\vy=\pm i \rho\vx$, since $(\skprod{\vx}{\pm i\rho\vx})^2=-\rho^2$. In fact,
$\pm i\rho \vx$ is a one-dimensional real subspace of $\C^N$ which parameterize $\vY=\vzero$. Hence, these problems
suggest to reformulate the complex case as a real-valued case. In the interest of space,  we will in this work only
consider the stability analysis for the real case and treat the complex case in a follow up paper.  

\section{Real Case}
In the real case the scalar product in \eqref{eq:Ynorm} is always real valued and so its product positive. Hence, the
Schatten $2-$norm of $\vY$ is for $\Norm{\vy}=1$ bounded by Cauchy-Schwarz
\begin{align}
  2\leq\sNorm{\vY}^2\leq 4.\label{eq:boundY}
\end{align}
Hence we can lower bound the stability constant in \eqref{eq:localRIPAlin} by a smallest singular value problem
\begin{align}
  \gam^2\geq \frac{1}{4}\min_{\Norm{\vy}=1} \! \Norm{\Alin(\vx\vy^T\!+\!\vy\vx^T)}^2 
  =\frac{1}{4}\min_{\Norm{\vy_{\!-}}=1} \Norm{\vM\vy_-}^2
  \label{eq:stabdef}
\end{align}
where we re parameterized $\vy$ as $\vy_-=[\vy_1 , -\vy_2^-]$.
Then indeed, there exists a linear mixing map $\vM$ such that with \eqref{eq:aijk} and \eqref{eq:Aijtrace}  we get
\begin{align}
  \Alin(\vY)&=\Alin(\vx\vy^T+\vy\vx^T)=
   \begin{pmatrix} 
       \vx_1*\vy_1^- \!+\! \vy_1*\vx_1^-\\
       \vx_1*\vy_2^- \!+\! \vy_1*\vx_2^-\\
       \vx_2*\vy_1^- \!+\! \vy_2*\vx_1^-\\
       \vx_2*\vy_2^- \!+\! \vy_2*\vx_2^-
     \end{pmatrix}\label{eq:AY}\\
  \setlength\arraycolsep{3pt}
&  \overset{\eqref{eq:Txij}}{=} 
\begin{pmatrix} 
  \vJ_1\vT_{1,\vxonet} \!\!    & \zero \\ 
  \vT_{1,\vxtwot}              &   \vT_{2,\minus\vx_1} \\
  \vR_{N\!-\!1}\vT_{1,\vxtwot} &   \vR_{N\!-\!1}\vT_{2,\minus\vx_1} \\
  \zero                        & \vJ_2\vT_{2,\minus\vxtwot} \!\!
\end{pmatrix}\!
  \begin{pmatrix} \vy_1\\ -\vy_2^-\end{pmatrix}\!= \!\vM\vy_-\notag
\end{align}
where we used the  time reversal matrix \eqref{eq:Hxij} and a  $(2\Ni-1)\times (2\Ni-1)$ intertwining matrix
\begin{align}
  \vJ_i\!=\!\id_{\!2\Ni\!-\!1}\! + \!\vR_{\!2\Ni\!-\!1}
  \label{eq:intertwining2}
\end{align}
for which holds $\vJ_i=\vJ_i^T$ and $\vJ_i^2=2\vJ$, since $\vR_{2\Ni-1}^T=\vR_{2\Ni-1}$ and
$\vR_{2\Ni-1}^2=\id_{2\Ni-1}$. Moreover, we can identify a Sylvester matrix without the zero row as
\begin{align}
  \vtS_-:=\vtS_{\vx_2^{-0},-\vx_1^0}=\begin{pmatrix} \vT_{1,\vx_2^-} & \vT_{2,-\vx_1} \end{pmatrix}\label{eq:Stilde}.
\end{align}
Hence, we can write \eqref{eq:stabdef} as a quadratic form by defining the positive
semi-definite matrix
\begin{align}
  \vM^T\!\vM &\!=\!
  \setlength\arraycolsep{1.5pt}
  \begin{pmatrix} 
    \!\begin{bmatrix}\! \vT_{1,\vxonet}^T\!\vJ_{_{\!1}} \\ \vzero\end{bmatrix} &
    \vtS^T_{-}& \vtS^T_{-}\vR & 
      \begin{bmatrix} \vzero \\ \vT_{2,\vxtwot}^T \!\vJ_{_{\!2}} \end{bmatrix} \!
  \end{pmatrix}\!\!
  \begin{pmatrix} 
    \begin{bmatrix}\! \vT_{1,\vxonet}\vJ_1 & \vzero\end{bmatrix} \\
    \vtS_{-}\\
    \vR\vtS_{-}\\
    \begin{bmatrix} \vzero &\vJ_{_{\!2}}\! \vT_{2,\text{-}\vxtwot}  \end{bmatrix}
  \end{pmatrix}\notag
  \\
  &= 
  \setlength\arraycolsep{2pt}
     \!\begin{bmatrix}\! \vT_{1,\vxonet}^T\!\vJ_{_{\!1}}^2 \vT_{1,\vxonet} & \vzero \\ 
       \vzero & \vT_{2,\text{-}\vxtwot}^T \!\vJ_{_{\!2}}^2 \vT_{2,\text{-}\vxtwot}
    \end{bmatrix}
     +
     2\vtS_{-}^T\vtS_{-}
     \label{eq:Bdef}
     \\
  &=
    2 
    \begin{pmatrix}\vD_1 & \zero \\ \zero & \vD_2\end{pmatrix}
    + 2\vS_{-}\vS^T_-
    = 2(\vD + \vW_{-})
\end{align}
where the block diagonal matrix $\vD$ and the Sylvester matrix product, given by 
\begin{align}
  \vW_{-}=\vS_{\vx_2^{-0},-\vx_1^0}^*\vS_{\vx_2^{-0},-\vx_1^{0}},\label{eq:Wminus}
\end{align}
are positive semi-definite by construction.  Note, we can add the zero row to the matrix in \eqref{eq:Stilde} to obtain
the square Sylvester matrix in the product, see also \eqref{eq:sylvesterrangestar}.  However, this Sylvester matrix is
distinct to $\vS$ in \eqref{eq:sylvesterx} by a time-reversal of $\vx_2$ and therefore denoted by $\vS_-$. The
time-reversal is hereby unavoidable, since the cross correlation sum in \eqref{eq:AY} involves all four vectors and
induces a time-reversal either on $\vx_1$ or $\vx_2$. This is in contrast to the injectivity result where we do not have
to deal with a sum  and hence not with an extra time-reversal \cite{WJPH17}.  The intertwining matrices generate in the
diagonals 
\begin{align}
  \begin{split}
 \vD_i\!=\! \vT_{i,\vxit}^T \!\vJ_{i}\vT_{i,\vxit}&=\vT_{i,\vxit}^T \!\vT_{i,\vxit}\!+\!
  \vT_{i,\vxit}^T\!\vR_{2N_i-1}\vT_{i,\vxit}\\
  &\overset{\eqref{eq:convTopelitzproduct}}{=}\vT_{i,i,\va_i} + \vH_{i,i,\vc_{i}}
\end{split}
\end{align}
a sum of the autocorrelation Toeplitz matrix and the autoconvolution Hankel matrix. Here, the elementary Hankel matrix
is given by $\vH_N=\vT_N\vR_N$.  The stability constant is then given, up to a constant scaling between
\eqref{eq:boundY} as the smallest eigenvalue problem 
\begin{align}
  \gam^2 \geq \frac{1}{2} \min_{\Norm{\vy}=1}
  \skprod{(\vD+\vW_-)\vy}{\vy}=\frac{1}{2}\lam_1(\vD+\vW_-)\label{eq:gamDWoptimization}.
\end{align}

\subsection{Local $2-$RIP}

We are now ready to proof the local restricted isometry property  for rank$-2$ matrices in the tangent space.
\begin{lemi}[Local $2-$RIP]\label{lem:localstab}
  Let $\vx_1\in\R^{\Lone}_{0},\vx_2\in\R^{\Ltwo}_{0,0}$ such that $\ux_1\nmid \ux_2^-$ and set $\vx=[\vx_1,\vx_2]$ with
  $N:=\Lone+\Ltwo$. Then  the linear map $\Alin$ in \eqref{eq:4maskmeasurements}   fulfills
  \begin{align}
    \Norm{\Alin(\vY)} \geq \gam \sNorm{\vY} \quad,\quad \vY\in T_\vx\label{eq:localrank1RIP}
  \end{align}
  where the lower bound satisfy
  \begin{align}
    \gam=\gam(\vtx_1,\vtx_2) &\geq \frac{1}{4N\sqrt{2}} \sqrt{\lam_2(\vW_{-})}
     \label{eq:gammalowerbound}
  \end{align}
  with $\vtx_i=\vx_i/\Norm{\vx}$ for $i\in\{1,2\}$.
\end{lemi}
\begin{proof}
Since the optimization problem \eqref{eq:gamDWoptimization} for $\gam$ is only dependent on the normalized vectors $\vtx$ we will
omit the tilde notation over  $\vx,\vx_1$ and $\vx_2$.  

The eigenvalue problem is a simultaneous eigenvalue problem and bounds like the \emph{dual Weyl inequality}
\cite[Thm.4.3.1]{HJ13} 
\begin{align}
  \lam_1(\vD+\vWminus)\geq  \lam_1(\vD)+\lam_1(\vWminus)\label{eq:dualweyl}
\end{align}
are not sufficient, since they separate the autocorrelation and crosscorrelation measurements, which both can vanish
separately but not simultaneously. However, we know much more about $\vWminus$, in fact $\vWminus$ has rank $N-1$ if the
polynomials $\ux_2^- (z)$ and $\ux_1(z)$ are co-prime, see also \eqref{eq:sylvesterrank}, and its one-dimensional
nullspace is spanned by
\begin{align}
  \vx_-=\begin{pmatrix}\vx_1 \\ \vx_2^-\end{pmatrix}.\label{eq:wnull}
\end{align}
Hence, we can project each normalized $\vy\in\R^N$ to  $\vx_-$ and its orthogonal normalized complement space
$\Wperp:=\set{\vw\in\R^N}{\vw\bot\vx_-,\Norm{\vw}_2=1}$ such that  
\begin{align}
  \vy=\alp\vx_- + \bet\vwperp\quad,\quad \alp,\bet\in\R,\vwperp\in\Wperp.\label{eq:ranktwoy}
\end{align}
Since $\Norm{\vx_-}=\Norm{\vwperp}=1$, it must hold $\Norm{\vy}_2^2= \alp^2\Norm{\vx_-}+\bet^2\Norm{\vwperp}=\alp^2+\bet^2=1$.
If we want to establish injectivity it would be enough to show that $\skprod{\vD\vx_-}{\vx_-}>0$, but to derive a
stability bound we need to show this  for all $\vy$ given by \eqref{eq:ranktwoy}, i.e.,
\begin{align}
  \frac{1}{2}\!\Norm{\vM\vy}^2&\!=\!\skprod{(\vD\!+\!\!\vWminus)\vy}{\vy} \!=\! \skprod{\vD\vy}{\vy}+\skprod{\vWminus\vy}{\vy}\\
  &\!=\! \skprod{\vD(\alp\vx_-\!+\!\bet \vwperp)}{\alp\vx_-\!+\!\bet\vwperp} \!+\! \bet^2\!\skprod{\vWminus\vwperp}{\vwperp}\notag.
\end{align}
Therefore we can lower bound the smallest eigenvalue  by a kind of pinching argument to
\begin{align}
  \lam_1&\!=\min_{\Norm{\vy}=1} \skprod{(\vD+\vWminus)\vy}{\vy}\label{eq:Dmin} \\
  &\!=\!\min_{\substack{\alp^2+\bet^2=1\\ \vw\in\Wperp}}
  \!\underbrace{\skprod{\vD(\alp\vx_-\!+\!\bet \vwperp)}{\alp\vx_-\!+\!\bet\vwperp}}_{=f_{\alp,\bet}(\vw)\geq 0} 
   + \bet^2\skprod{\vWminus\vwperp}{\vwperp}\notag\\
  &\geq\min_{\alp^2+\bet^2=1}\left( \bet^2\lam_2(\vWminus) +\min_{\vwperp\in\Wperp}f_{\alp,\bet}(\vwperp)
  \right)\notag.
\end{align}
Since $\vD=\vD^T$ and the scalar product is real valued we get
\begin{align}
  f_{\alp,\bet}(\vwperp\!)\!=\! \alp^2\!\skprod{\vD\vx_-}{\!\vx_-\!}\!+\!\bet^2\!\skprod{\vD\vwperp}{\!\vwperp}
  \!+\!2\alp\bet\!\skprod{\vD\vx_-}{\!\vwperp}\label{eq:Dyy}.
\end{align}
The first term can be lower bounded universally for all $\Norm{\vx}_2^2=1=\Norm{\vx_1}^2+\Norm{\vx_2}^2$ by 
\begin{align}
    \skprod{\vD\vx_-}{\vx_-}&\overset{\eqref{eq:wnull}}{=}\skprod{\vD_1\vx_1}{\vx_1} +
    \skprod{\vD_2\vx_2^-}{\vx_2^-}\notag\\
    &\overset{\eqref{eq:Bdef}}{=}\frac{1}{2}\Big(\Norm{\vJ_1(\vx_1^-*\vx_1)}^2_2 +  \Norm{
    \vJ_2(\vx_2*\vx_2^-)}^2_2\Big) \notag\\
    &= \frac{1}{2}\Big( \Norm{2\va_1}_2^2 + \Norm{2\va_2}_2^2 \Big)=2(\Norm{\va_1}_2^2 + \Norm{\va_2}_2^2) \notag\\
   &\geq
   2(\Norm{\vx_1}_2^4+\Norm{\vx_2}_2^4) \overset{\eqref{eq:lagrange}}{\geq} 1
 \label{eq:dwowo}
\end{align}
where the first inequality follows from the fact that $\|{\vx_i*\vxict\|}^2_2\geq \Norm{\vx_i}_2^4$. Unfortunately, we
can not derive such a lower estimation for $\skprod{\vD\vwperp}{\vwperp}$ since arbitrary cross correlations are not
invariant to time-reversal. To see this, choose  $\vwperpone=(1,1)=\vx_2$ and $\vwperptwo=(1,-1)=\vx_1$. Then
$\vx_i\bot\vwperpi$ and hence $\vx_-\perp\vwperp$ but $\vwperpi*\vx_i^-+ (\vwperpi*\vx_i^-)^-=\zero$. Also
$\ux_1\nmid\ux_2$ and $\ux_2=\ux_2^-$.  If we omit the non-negative part $\bet^2\skprod{\vD\vwperp}{\vwperp}\geq 0$ we
only have to lower bound the third term in \eqref{eq:Dyy}. Since $f_{\alp,\bet}\geq 0$ this gives with \eqref{eq:dwowo}
the non-zero lower bound 
\begin{align}
  \min_{\vwperp\in\Wperp}\!\!f_{\alp,\bet}(\vwperp) &\!\geq\! \max\{0,\alp^2 - 2
    |\alp\bet|\max_{\vwperp\in\Wperp}\skprod{\vD\vx_-}{\vwperp}\},\notag
\end{align}
where the scalar product can be split with $\vwperp\!=[\vwperpone,\vwperptwo]$ in
\begin{align}
  \begin{split}
  |\!\skprod{\vD\vx_-}{\!\vwperp}\!| \!&= |\skprod{\vD_1\vx_1}{\vwperpone} +\skprod{\vD_2\vx_2^-}{\vwperptwo}|\\
   &\leq |\skprod{\vD_1\vx_1}{\vwperpone}| +|\skprod{\vD_2\vx_2^-}{\vwperptwo}|\\
   &=  |\skprod{\vJ_1\vT_{1,\vx_1^-} \vx_1}{\vT_{1,\vx_1^-} \vwperpone}|\\
  &\quad\quad \quad+|\skprod{\vJ_2\vT_{2,\vx_2} \vx_2^-}{\vT_{2,\vx_2} \vwperptwo}|\\
  &=2 (|\skprod{\va_1}{\vx_1^-*\vwperpone}| + |\skprod{\va_2}{\vx_2*\vwperptwo}|)
\end{split}
\end{align}
such that by using Cauchy-Schwarz in both terms we get
\begin{align}
  |\!\skprod{\vD\vx_-}{\!\vwperp}\!| \! &\leq \!2 (\Norm{\va_1}\Norm{\vx_1^-\!*\vwperpone} \!+\!
  \Norm{\va_2}\Norm{\vx_2*\vwperptwo})\\
  \intertext{and by applying the Young inequality}
  \begin{split}
    &\leq\!2\Big( \Norm{\vx_1}_1^2 \underbrace{\Norm{\vx_1}}_{\leq 1}\underbrace{\Norm{\vwperpone}}_{\leq 1} + \Norm{\vx_2}_1^2
    \underbrace{\Norm{\vx_2}}_{\leq 1}\underbrace{\Norm{\vwperpone}}_{\leq 1}\Big)\\
    &\leq\! 2\Big( \Lone\Norm{\vx_1}^2 + \Ltwo\Norm{\vx_2}^2\Big)\leq 2\max\{\Lone,\Ltwo\}\\
    &\leq\! 2(N-1).
  \end{split}\notag
\end{align}
where the last steps follows since $N_1,N_2\geq 1$ and $N=N_1+N_2$.
If $\Norm{\vx_1}=\Norm{\vx_2}$ or $\Lone=\Ltwo$ then the bound is $N$.  Hence 
\begin{align}
  \min_{\vwperp} f_{\alp,\bet}(\vwperp)\geq |\alp|\max\{0,|\alp| -4(N-1)|\bet|\}  \label{eq:dzero}
\end{align}
which only gives a strict positive bound if $|\bet|\leq \frac{|\alp|}{M}$ with $M=4N-4$. Hence we get
\begin{align}
  \lam_1 \geq  \min_{\substack{|\alp^2|+|\bet|^2=1\\ M|\bet|\leq |\alp|}} 
\Big[|\bet|^2 \lam_2(\vWminus)+|\alp|^2 -M|\bet||\alp|\}\Big]\notag\label{eq:lagrangeobjective}
\end{align}
Since $0<\lam_2(\vWminus)=\sig_2^2(\vS_-)\leq 1$ as shown in \lemref{lem:sigclassic} since $\vS_-$ has rank $N-1$, one
can show by using Lagrange multiplier that $M|\bet|=|\alp|$ with $|\bet|^2=1/(1+M^2)$ yields the minimum in
\eqref{eq:lagrangeobjective} given by
\begin{align}
  \lam_1 \geq \frac{1}{1+M^2} \lam_2(\vW_-).
\end{align}
By using $M^2+1\leq (4N)^2$ this yields to the lower bound 
\begin{align}
  \gam(\vx_1,\vx_2)&\geq\frac{1}{\sqrt{2}} \frac{1}{4N} \sqrt{\lam_2(\vWminus)}
  .\label{eq:gaminlam2}
\end{align}
\end{proof}

\subsection{Stability of the SDP}

Since we have the exact dual certificate, we only have to exploit in the noisy case, that any solution produces a
mean-square error which scales with the noise-power.  We adapt the proof techniques in \cite[Lem.4]{DH12} to derive the
stability result for the denoised SDP.
\renewcommand{\vxgt}{\vx}
%
%%%%%%%%%%%%%%% Inecact Lemma:Start %%%%%%%%%%%%%%
\begin{thm}[Stability of the SDP: Real Case]\label{thm:stability}
  Let $\Alin$ be the linear map defined in \eqref{eq:4maskmeasurements}. Let $\vxgt=[\vx_1,\vx_2]$ be the ground truth
  signal with  $\vx_1\in\R_{0}^{\Lone}$ and $\vx_2\in\R_{0,0}^{\Ltwo}$. If $\hat{\vX}$ is a solution of the noisy SDP
  \eqref{eq:SDPnoisy}, given by the noisy observation $\vb=\Alin(\vx\vx^*)+\vn$ for any $\vn\in\C^{4N-4}$, then it holds
  \begin{align}
    \sNorm{\hat{\vX}-\vx\vx^*} \leq C\Norm{\vn}
  \end{align}
  with $N=\Lone+\Ltwo$ and the stability constant
  \begin{align}
    C\!\leq\! 2N\!\left( \frac{23N^2}{\lam_2(\vW)\sqrt{\lam_2(\vWminus)}}\!+\!\frac{1}{\lam_2(\vW)}\! 
    +\! \frac{4\sqrt{2}}{\sqrt{\lam_2(\vWminus)}}\right)\label{eq:CinW}
  \end{align}
  where $\vW$ and $\vWminus$ are given in \eqref{eq:dualcertif} and \eqref{eq:Wminus} for $\vtx_i\!=\!\vx_i/\Norm{\vx}$. 
\end{thm}
%%%%%%%%%%%%%%% Inexact Lemma:End %%%%%%%%%%%%%%
%
\begin{proof}
  Let $\hat{\vX}\mgeq \zero$ and $\|\Alin(\hat{\vX})-\vb\|\leq \Norm{\vn}$, i.e., $\hvX$ is a solution of the noisy SDP
  \eqref{eq:SDPnoisy}. For $\vH:=\hvX-\vx\vx^*$ we get then for the residual
  \begin{align}
    \begin{split}
    \Norm{\Alin(\vH)}&=\|\Alin(\hvX)-\Alin(\vx\vx^*)+\vb-\vb\|\\
    &\leq \|\Alin(\hvX)-\vb\| +    \|\Alin(\vx\vx^*)-\vb\| \leq 2\Norm{\vn}. 
  \end{split}\label{eq:residualH}
  \end{align}
  Similarly we get for the projection onto the exact dual certificate by using Cauchy-Schwarz
  and the fact that $\vW=\Alin^*(\vome)$, as shown in \eqref{eq:AlamW}, 
  \begin{align}
    \begin{split}
    |\sprod{\vH}{\vW}|&=|\sprod{\Alin(\vH)}{\vome}|
        \leq \Norm{\Alin(\vH)}_2 \Norm{\vome}_2\\
        &\leq 2\Norm{\vn} \norm{\vome}_1.
      \end{split}\label{eq:absHW}
  \end{align}
  Since $\vome=[\va_1^{(2)},-\va_{1,2},-\va_{2,1},\va_2^{(1)}]$ in
  \eqref{eq:AlamW}, we  get  with the Young inequality
  \begin{align}
    \begin{split}
      \Norm{\vome}_1&=\Norm{\va_1^{(2)}}_1+\Norm{\va_{1,2}}_1 +\Norm{\va_{2,1}}_1 +\Norm{\va_2^{(1)}}_1\\
      &\leq\Norm{\vx_1*\vxonect}_1 + 2\Norm{\vx_1*\vx_2^-}_1+ \Norm{\vx_2*\vxtwoct}_1\\
                &\leq \Norm{\vx_1}_1^2+ 2\Norm{\vx_1}_1\Norm{\vx_2}_1 + \Norm{\vx_2}_1^2 \\
    &= (\Norm{\vx_1}_1 + \Norm{\vx_2}_1)^2
    = \Norm{\vx}_1^2
    \leq N\Norm{\vx}_2^2=N.
  \end{split}
  \end{align}
  This gives for the exact dual certificate projection \eqref{eq:absHW} on  $T=T_{\vx}$ and $T^{\bot}$ 
  \begin{align}
    2N\Norm{\vn} &\geq |\skprod{\vH}{\vW}| =|\skprod{\vH_{T}}{\vW_{T}} + \skprod{\vH_{T^\bot}}{\vW_{T^\bot}}|\notag\\
    &=\tr(\vU\vH_{T^\bot}\vU^T \vD_{\vlam} ) \geq \lam_2(\vW) \sNorm{\vH_{T^\bot}}_1\label{eq:htbotbound} 
  \end{align}
  where  $\vW_T=0$ and $\vW_{T^\bot}=\vU^T\vD_{\vlam}\vU\succ 0$ since the diagonal matrix $\vD_{\vlam}$ contains the
  positive eigenvalues $\lam_k$ of $\vW_{T^\bot}$ which must be all larger than the smallest non-zero eigenvalue of
  $\vW$, i.e.  $\lam_1(\vW_{T^\bot})=\lam_2(\vW)$. Moreover, $\vH_{T^{\bot}}=\hvX_{T^\bot}$ is positive semi-definite
  since $\hvX$ is by definition.  Hence the last inequality follows, see also \lemref{lem:tracelowerbound}.  To upper
  bound the norm of $\vH$ we therefore only need to upper bound the norm of  $\vH_T$, which is given by the local
  stability on $T$ in \lemref{lem:localstab}, 
  \begin{align}
    \gam\sNorm{\vH_T} &\overset{\eqref{eq:localrank1RIP}}{\leq} \Norm{\Alin(\vH_T)}=\Norm{\Alin(\vH-\vH_{T^\bot})} \label{eq:boundHT}\\
     &\leq \Norm{\Alin(\vH)} + \Norm{\Alin(\vH_{T^\bot})}
     \overset{\eqref{eq:residualH}}{\leq} 2\Norm{\vn}  +  \Norm{\Alin(\vH_{T^\bot})}_1 \notag
  \end{align}
  Now, we need an upper bound of $\Norm{\Alin(\vX)}_1$ for any Hermitian $\vX$, which is given by the \emph{Hölder
  inequality} for the \emph{Schatten norms} with $p=1$ and $q=\infty$, see \cite[Thm.2]{Bau11}, and \eqref{eq:4maskmeasurements} as
  \begin{align}
    \begin{split}
      \Norm{\Alin(\vX)}_1 &= \sum_{j\geq i=1}^2 \! \!\sum_{k=0}^{\Ni\!+\!\Nj\!-\!2} \!\!| \tr(\vA_{i,j,k} \vX)| \\
      &\leq\sum_{i\leq j,m} \sNorm{\vA_{i,j,k}}_{\infty} \sNorm{\vX}_1 = M\sNorm{\vX}_1
    \end{split}\label{eq:AlinXone}
  \end{align}
  where the last equality follows from the observation
  \begin{align}
    \snorm{\vA_{i,j,k}}^2_{\infty}\overset{\eqref{eq:inftynorm}}{=}\lam_{N}(\vA_{i,j,k}\vA_{i,j,k}^*)
    =\lam_{N} (\vT_{i,j}^{(k)}(\vT_{i,j}^{(k)})^T)=1\notag
  \end{align}
  for all $i,j,k$, since the elementary Toeplitz matrices $\vT_{i,j}^{(k)}$ generate diagonal matrices having $k$ ones
  on the diagonal and the rest zero. Using \eqref{eq:boundHT} and \eqref{eq:AlinXone} with $\vX=\vH_{T^\bot}$ 
  gives us 
  \begin{align}
    \sNorm{\vH_T}& \!\leq\! \frac{2}{\gam}\!\left(\!
    \Norm{\vn} \! +\! \frac{M}{2}\! \sNorm{\vH_{T^{\bot}}\!}_1\!\!\right) \!\!\overset{\eqref{eq:htbotbound}}{\leq}\!
    \frac{2\!\Norm{\vn}}{\gam}\!\left(\!\!1 \!+\! \frac{NM}{\lam_2(\vW)} \!\right)\!.\label{eq:HTnorm} 
  \end{align}
  Using \eqref{eq:HTnorm} and \eqref{eq:htbotbound} with the triangle inequality yields the error 
  \begin{align}
    \sNorm{\vH}&\leq \sNorm{\vH_T}+\sNorm{\vH_{T^\bot}}
    \leq \sNorm{\vH_T}+\sNorm{\vH_{T^\bot}}_1\\
    &\leq \Norm{\vn} \cdot\underbrace{\left( \frac{2}{\gam} + \frac{2NM}{\gam\lam_2(\vW)} +
    \frac{2N}{\lam_2(\vW)}\right)}_{=C}\notag.
  \end{align}
  By using the bound \eqref{eq:gammalowerbound} of \lemref{lem:localstab} and $4\sqrt{2}NM\leq 16\sqrt{2}N^2\leq
  23 N^2$ we get 
  \begin{align}
    C&\leq 2N\!\left( \frac{23N^2}{\lam_2(\vW)\sqrt{\lam_2(\vWminus)}}\!+\!\frac{1}{\lam_2(\vW)}\! +\!
    \frac{4\sqrt{2}}{\sqrt{\lam_2(\vWminus)}}\right)
    \notag
  \end{align}
\end{proof}

\subsection{Universal Stability Bound via Zero Structure}

To obtain a universal bound for the stability constant we need more constraints on the structure of the zeros of $\ux_1$
and $\ux_2$. It follows, that a universal stability can be obtained if the zeros fulfill a minimal separation. Let us
define the \emph{zero separation} of $\ux_1$ and $\ux_2$ by 
\begin{align}
  \del=\min_{l,k}| \alp_k-\bet_l|\label{eq:zeroseparationdel},
\end{align}
and the zeros between $\ux_1^-$ and $\ux_2$ by
\begin{align}
  \del_-=\min_{l,k}| \frac{1}{\cc{\alp_k}}-\bet_l|\label{eq:zeroseparationdelminus},
\end{align}
see also \cite[p.186]{Zip93}.
Then we can proof the following.
\begin{thm}\label{lem:lam2bound}
  Let $\vx_1\in\R^{\Lone}_{0}$ and $\vx_2\in\R^{\Ltwo}_{0,0}$ such that $N_1,N_2\geq 2$ and the zeros $\alp_k$ and
  $\bet_k$ of their $z-$transforms have a separation $\del,\del_->0$. Then the stability constant of the denoised
\eqref{eq:SDPnoisy}
  problem is upper bounded by 
  \begin{align}
  C\!\leq\!  48N^3 (\del_-\del^2)^{-N_1 N_2} \!\cdot|\tx_{1,0}^3\tx_{2,N_2\!-\!1}|^{-N_1}\!\cdot
  |\tx_{2,0}|^{-3N_2}.\label{eq:Cupperbound}
  \end{align}
\end{thm}
\begin{figure}[t]
  \centering
  \includegraphics[scale=0.4]{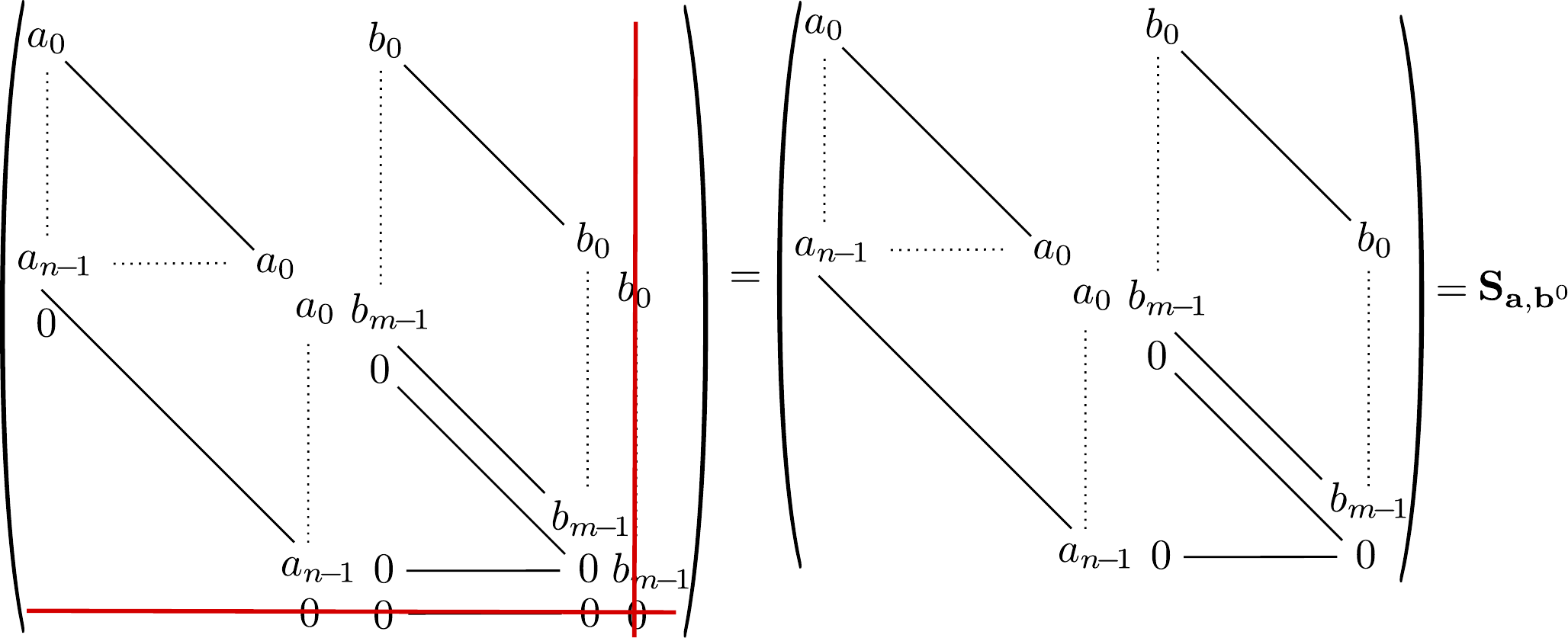}
  \caption{Principal Sylvester submatrix of $\vS_{\va^0,\vb^0}$.}\label{fig:S0}
\end{figure}
\begin{remark}
  For $N_1=1$ or $N_2=1$ the polynomials have no zeros and hence the zero separation \eqref{eq:zeroseparationdel} is not defined. 
\end{remark}
\noi{\it Proof. } Since the dual certificate $\vW=\vS_{\vx_2^0,-\vx_1^0}^*\vS_{\vx_2^0,-\vx_1^0}$ for any $\vx_1$ and
$\vx_2$ is a matrix product of Sylvester matrices \eqref{eq:sylvesterx} for which we can control the smallest singular
value \appref{app:sylvestereigenvalue} we also can control the smallest non-zero eigenvalue, since
\begin{align}
  \lam_k(\vW)=\lam_k(\vS^*\vS)=\lam_k(\vS\vS^*)=\sig_k^2(\vS) \label{eq:WSlamsig}
\end{align}
for $k\in\{1,\dots, N\}$.
Lets define $\vS_0=\vS_{\vx_2,-\vx_1^0}$, the $N-1\times N-1$ principal submatrix of $\vS=\vS_{\vx_2^0,-\vx_1^0}$ by
deleting the last column and row (handle)
\begin{align}
  \vS=
  \left(
  \begin{array}{ccc}
    \multicolumn{2}{c|}{\multirow{2}{*}{\raisebox{-.1\normalbaselineskip}[2pt]{\large$\vS_0$}}} & \zero \\
    \multicolumn{2}{c|}{} & -\vx_1\!\!\\
    \cline{1-2}
    \rule{0pt}{2.2ex}
    \zero^* & \zero^*& 0 
  \end{array}  
  \right).
\end{align}
see \figref{fig:S0}.
Since the lower part of the handle is zero the multiplication with its adjoint creates a zero matrix and the non-zero
upper part of the handle creates a positive semi-definite rank$-1$ matrix
\begin{align}
\vS\vS^*\!=\! 
\left(
\begin{array}{ccc}
  \multicolumn{2}{c|}{\multirow{2}{*}{\raisebox{-.1\normalbaselineskip}[2pt]{\large$\vS_0\vS_0^*$}}} & \zero \\
  \multicolumn{2}{c|}{} & -\vx_1\!\!\!\!\!\\
  \cline{1-2}
  \rule{0pt}{2.2ex}
  \zero^* & \!\!-\vx_1^*\!& 0 
\end{array}  
\right)
\!+\!
\left(
    \begin{array}{ccc}
      \multicolumn{2}{c|}{\multirow{2}{*}{\raisebox{-.6\normalbaselineskip}[-2pt]{$\begin{matrix} \vzero & \vzero\\
        \vzero &
        \vx_1\vx_1^T\end{matrix}$}}} & \zero
        \rule{0pt}{2.2ex}
      \\
      \multicolumn{2}{c|}{} & \zero\\
      \cline{1-2}
      \rule{0pt}{2.2ex}
      \zero^*& \zero^*& 0 
    \end{array}  
\right).
\end{align}
Obviously we have $\lam_1(\vS\vS^*)=0$ and
\begin{align}
  \lam_2(\vS\vS^*)=\lam_1\Big(\vS_0\vS_0^* +
  \renewcommand{\arraystretch}{0.8}
  \begin{pmatrix}\zero \\ \vx_2\end{pmatrix}
  \setlength\arraycolsep{0.8pt}
  \begin{pmatrix}\zero^* &
    \vx_2^*\end{pmatrix} 
  \Big).
\end{align}
Since both matrices are positive semi-definite we get by the dual Weyl inequality \eqref{eq:dualweyl}
\begin{align}
  \lam_2(\vW)&\overset{\eqref{eq:WSlamsig}}{=} \lam_2(\vS\vS^*)\geq \lam_1(\vS_0\vS_0^*)
  =\sig_1^2(\vS_{0})=\sig^2.
\end{align}
Similar, we get for $\vWminus=\vS^*_-\vS_-$ with $\vS_-=\vS_{\vx_2^{-0},-\vx_1^0}$
\begin{align}
  \lam_2(\vWminus)\geq \sig_1^2(\vS_{\vx_2^-,-\vx_1^0})=\sig_-^2.
\end{align}
Inserting this in \eqref{eq:CinW} gives
\begin{align}
  C\leq 2N\left(\frac{23N^2}{\sig^2 \sig_-} \!+\! \frac{1}{\sig^2}\! +\!\frac{4\sqrt{2}}{\sig_-}  \right)
  \!=\!2N\frac{23N^2\! +\! \sig_-\! +\! 4\sqrt{2}\sig^2}{\sig^2\sig_-}.\notag
\end{align}  
Using the bounds in \lemref{lem:sigclassic} we get
\begin{align}
  1\geq \sig &\geq |\tx_{1,0}|^{\Lone} |\tx_{2,0}|^{\Ltwo} \del^{\Lone\Ltwo}\\
   1\geq\sig_-&\geq  |\tx_{1,0}|^{\Lone} |\tx_{2,\Ltwo-1}|^{\Ltwo} \del_-^{\Lone\Ltwo}
\end{align}
and with $\sig_-+4\sqrt{2}\sig^2\leq 7\leq N^2$ for $N\geq 4$, this yields to
\begin{flalign}
&&  C\leq 48N^3\sig^{-2}\sig_-^{-1} && \qed\notag
\end{flalign}
%
%
%%%%% end proof

  \section{Conclusion}
  
We have shown that blind deconvolution with additional autocorrelation measurements is  stable over the reals against
additive noise if the zeros of the corresponding input polynomials are well separated and the first and last
coefficients are dominating. The last observation is well known in filter theory and signal processing. If for example
the first coefficient contains more than half the energy of the vector, then the polynomial respectively $z-$transform
corresponds to a minimum phase filter (System) having all its zeros inside the unit circle. Vice versa, if this holds
for the last coefficient, this corresponds to maximum phase filter, having all its zeros outside the unit circle.
Although, the stability bound decays exponentially in the dimension, it gives deeper and provable  insight how a zero
structure leads to good deconvolution or phase retrieval performance.  Moreover, the algorithm is convex and  with
modern semi-definite programing algorithms sufficiently solvable. However, since the program proofs to obtains a unique
(stable) solution up to a global phase, it is plausible that also non-convex relaxations, such as Wirtinger flow or the
\emph{direct zero testing method} (DiZeT), as introduced in \cite{WJH17b,WJH17c} and \cite{CWH17}, might perform stable,
as has been empirical observed.

\section*{Acknowledgement}

The authors would like to thank Ahmed Douik, Richard Kueng and Peter Jung for many helpful discussions.  The work of
Philipp Walk was supported by the German Research Foundation (DFG) under the grant WA 3390/1 and the one of Babak
Hassibi was supported in part by the National Science Foundation under grants CNS-0932428, CCF-1018927, CCF-1423663 and
CCF-1409204, by a grant from Qualcomm Inc., by NASA’s Jet Propulsion Laboratory through the President and Director’s
Fund, by King Abdulaziz University, and by King Abdullah University of Science and Technology.

  \section*{References}
  \printbibliography[heading=bibintoc]

  \appendices

  \section{Sylvester Matrix}\label{app:sylvester}%
   
The vectors $\va\in\C^{n+1}_0,\vb\in\C^{m+1}_0$ generate (not necessarily monic) polynomials of order $n$ respectively
$m$ 
\begin{align}
  \begin{split}
  z^{n}\uA(z)&\!=\sum_{k=0}^{n} a_{n-k} z^k=a_0 \Pro_{k=1}^{n} (z-\alp_k)\!=\!\uareci(z) \\
  z^{m}\uB(z)&\!=\sum_{l=0}^{m} b_{m-l} z^l=b_0 \Pro_{l=1}^{m} (z-\bet_l)=\ubreci(z),
\end{split}
\label{eq:uaminusfac}
\end{align}
where $\alp_k$ and $\bet_l$ denote the zeros of the polynomial $\uareci$  respectively $\ubreci$.  Hence the zeros of
the $z-$transform \eqref{eq:ZformX} equal the zeros of the polynomials.  By commutativity we assume w.l.o.g. $n\leq m$.
Then the $N\times N$ \emph{Sylvester matrix\index{matrix!Sylvester}} of the polynomial $\uareci$ and $\ubreci$ is with
$N=n+m$ given by
\begin{align}
\vspace{1cm}\notag\\
\hspace{-0.2cm}  \vS_{\va,\vb}\!=\!
   \setlength\arraycolsep{2pt}
  \left(\begin{array}{cccc|cccccc}
    \tzm{u1a}a_{0} & 0        & \dots & 0\tzm{u1b}     & \tzm{u2a}b_{0}& \dots &\dots &\dots &0\tzm{u2b} \\
    a_{1}          & a_{0}    & \dots &  0             & b_{1}         &\dots& \dots &\dots &0\\
    \vdots         &          & \diagdown& \vdots      & \vdots        & \ddots && & \vdots  \\
    \vdots         & \vdots   &          & \vdots      & b_{n}         &\dots & b_{0} & \dots & 0\\
    a_{m}          & a_{m-1}  & \dots    & a_{m-(n-1)} &  \vdots       &\diagdown & &\ddots &\vdots \\
    0              & a_{m}    & \dots    & a_{m-(n-2)} & 0             &   & b_{n} & & b_{0}\\
    \vdots         &          &\diagdown & \vdots      & \vdots        &   & &\diagdown&\vdots\\
    0              & 0        & \dots    & a_{m}       & 0             & \dots & 0 & \dots& b_{n} 
  \end{array}\right)\label{eq:sylvestermatrix}
\end{align}
where the first $n$ columns are down shifts of the vector $\vb$ and the last $m$ columns shifts of the vector $\va$,
see for example \cite[Sec.VII]{JH16} or \cite[Def.7.2]{GCL92} (here they define the transpose version and for the
reciprocal polynomials)
The \emph{resultant} of the polynomial $\uareci$ and $\ubreci$ is the determinant
of the Sylvester matrix $\vS_{\va,\vb}$. \namen{Sylvester} showed that the two
polynomials have a common factor (non-trivial polynomial, i.e., not a constant) if and only if $\det(\vS_{\va,\vb})\not=0$,
which is equivalent of having full rank, i.e.,  $\rank(\vS_{\va,\vb})=n+m$.

\subsection{Singular Value Estimates of Sylvester Matrices}\label{app:sylvestereigenvalue}

Let us conclude this section with a eigenvalue lower bound. If $\va\in\C^{n}_{0,0}$ and
$\vb\in\C^{m}_{0,0}$ then the first and last coefficient of $\va$ and $\vb$ are non-vanishing. Hence,
$\ua,\ub$ are full degree polynomials with non-vanishing zeros (roots).
If $\uareci$ and $\ubreci$ are co-prime, i.e., share no common zeros, then their resultant,  given by
\begin{align}
  \det \vS_{\va,\vb}\!=\!  b_0^{n\!-\!1}\Pro_{l=1}^{m\!-\!1}\uareci(\bet_l)
  &\!\!\overset{\eqref{eq:uaminusfac}}{=}
  a_{0}^{m\!-\!1} b_{0}^{n\!-\!1}\Pro_{l=1}^{m-1} \Pro_{k=1}^{n-1} (\bet_l\! -\!\alp_k)\label{eq:detsab}
\end{align}
is non-vanishing and the $N\!-\!2\times N\!-\!2$ Sylvester matrix $\vS_{\va,\vb}$ has full rank, see for example
\cite[Thm.1]{BL09} and \cite[Sec.9.2]{Zip93}.  
If we  append a zero to the vector $\va$ denoted by $\va^0$, see \eqref{eq:xaddzero}, the corresponding polynomial
$\ua^0(z)$ has order $n$ 
\begin{align}
  \ua^0(z)&=\sum_{k=0}^{n} a_{n\!-\!k}^0 z^k = z\sum_{k=0}^{n-1} a_{n\!-\!1\!-\!k} z^k =z \Pro_{k=1}^{n} (z-\alp_k)
\end{align}
and the zeros $\alp_k$ of $\ua$ are the non vanishing zeros of $\ua^0$. We denote the vanishing zero by $\alp^{0}_{n}=0$.
If we do the same for $\ubreci$ then the only common factor of  $\ua^0$ and $\ub^0$ is $z$, and the determinant of the
$N\times N$ matrix $\vS_{\va^0,\vb^0}$ vanishes by \eqref{eq:detsab} since $\alp_n^{0}=\bet_m^{0}=0$.  However, if we consider
the principal submatrix $\vS_{\va^0,\vb}$, which is given by erasing the last row and last column in $\vS_{\va^0,\vb^0}$,
then the determinant is non vanishing since $\ua^0(z)$ and $\ub(z)$ share no common factor and hence
\begin{align}
  \rank(\vS_{\va^0,\vb^0})=N-1 \ \LRA\ \forall l,k\colon 0\not=\alp_k\not=\bet_l\not=0\label{eq:sylvesterrank}
\end{align}

Let us now use \eqref{eq:detsab} to show a lower bound for the smallest singular value of the Sylvester matrix.
\begin{lemi}\label{lem:sigclassic}
Let $\va\in\C^{n+1}_{0}$ and $\vb\in\C^{m+1}_{0}$ for $n,m\geq 1$ with $\Norm{\va}^2+\Norm{\vb}^2=1$. If the zeros of $\ua$ and $\ub$
have a zero separation $\del>0$ as defined in \eqref{eq:zeroseparationdel}, then it holds
\begin{align}
  1\geq\!\frac{\sig_{n+m}(\vS_{\va,\vb})}{\sqrt{n+m-1}}\!\geq\!  \sig_1(\vS_{\va,\vb})\geq |a_0|^m |b_0|^n \del^{nm}.
\end{align}
\end{lemi}
\noi {\it Proof. }
Since $\del>0$ the polynomials $\ua$ and $\ub$ are coprime.  Hence the rank of $\vS=\vS_{\va,\vb}$ is full and all
singular values $\sig_k>0$. It is known that the absolute value of the determinate $\vS=\vS_{\va,\vb}$ is the product of
its singular values
\begin{align}
  |\det\vS|= \Pro_{k=1}^{n+m} \sig_k (\vS).
\end{align}
Hence we get for the smallest singular value by \cite[Cor.3]{Gun07}
\begin{align}
  \sig_1(\vS)\geq  \left(\frac{n+m-1}{\sNorm{\vS}^2}\right)^{\frac{n+m-1}{2}}\Betrag{\det(\vS)}.\label{eq:gun}
\end{align}
Since the Sylvester matrix is given by
\begin{align}
  \vS=\vS_{\va,\vb}=\begin{pmatrix} \vT_{m,\va} & \vT_{n,\vb}\end{pmatrix}
\end{align}
we get for the Frobenius norm
\begin{align}
  \sNorm{\vS}^2&=\tr(\vS^*\vS)=\tr\begin{pmatrix} \vT_{m,\va}^*\vT_{m,\va} & \vT_{m,\va}^*\vT_{n,\vb}\\
    \vT_{n,\vb}^*\vT_{m,\va} & \vT_{m,\va}^* \vT_{m,\va}\end{pmatrix}\notag\\
  &\!\!\overset{\eqref{eq:convTopelitzproduct}}{=}\tr(\vT_{m,\va*\svact})+\tr(\vT_{n,\vb*\svbct})\notag\\
  &=
  m\Norm{\va}_2^2 + n\Norm{\vb}_2^2 \leq \max\{n,m\}(\Norm{\va}^2 + \Norm{\vb}^2 \notag\\
  &\leq \max\{n,m\}\leq n+m-1.\label{eq:frobeniusnormS}
\end{align}
Since $\vS$ has full rank we get $\sNorm{\vS}^2=\sum_{k=1}^{n+m} \sig_k^2\geq (n+m)\sig_1^2$ and hence with
\eqref{eq:frobeniusnormS} the upper bound for $\sig_1$. Furthermore, $n+m-1$ is also an upper bound for $\sig_{max}^2$.
Moreover, we get with \eqref{eq:gun} the lower bound 
\begin{flalign}
  &&\sig_1(\vS)\geq |\det(\vS)| \geq |a_0|^m |b_0|^n \del^{nm}&& \qed\notag
\end{flalign}
%
% ende proof
%

\renewcommand{\alp}{\ensuremath{\alpha}}
\renewcommand{\valp}{\ensuremath{\bm \alpha}}
\subsection{Alternative bound for the smallest singular value of the Sylvester Matrix}

If there is separate knowledge of the zero structure of $\ux_1$ and $\ux_2$ this can lead to tighter bounds. 
\begin{lemi}
  Let $\va\in\C^{n+1}_0,\vb\in\C^{m+1}_0$ as in \lemref{lem:sigclassic}, then
  \begin{align}
    \sig_1(\vS_{\va,\vb})&\geq \frac{\del^{N\!-\!1}|a_0b_0|}{N^2}\cdot \frac{\min\{(a/2)^{n-1},(b/2)^{m-1}\}}{a^{1-N}+b^{1-N}}
  \end{align}
  where $N=n+m$, $\del>0$ and
  \begin{align}
  a&=\frac{|a_0|}{|a_0|+1}\quad,\quad b=\frac{|b_0|}{|b_0|+1}.
\end{align}
\end{lemi}
\noi{\it Proof. }
Let $\va\in\C_0^{n+1}$ and $\vb \in\C_0^{m+1}$ for $n,m\in\N$. Then it holds for the Sylvester matrix and any
$\lam\in\C$
\begin{align}
\setlength\arraycolsep{2pt}
  &\begin{pmatrix}1  & \lam & \dots \ \lam^{N\!-\!1}\end{pmatrix} \vS_{\va,\vb}\label{eq:lamnotation}\\
  & \quad = 
  \setlength\arraycolsep{1.5pt}
  \begin{pmatrix}\ua(\lam) & \lam \ua(\lam) & \dots \lam^n \ua(\lam) & \lam^{n\!+\!1} \ub(\lam) \ \dots \
     \lam^{N\!-\!1} \ub(\lam)\end{pmatrix}\notag
\end{align}
with $N=n+m$.
Let us introduce the zero vector 
\begin{align}
  \vzet = \begin{pmatrix}\vbet \\ \valp\end{pmatrix}\in\C^N
\end{align}
where $\alp_k$ and $\bet_l$ are the zeros of $\ua(z)$ respectively $\ub(z)$ as in \eqref{eq:uaminusfac}.
Let us define the $N\times N$ Vandermonde matrix
\begin{align}
  \vV_{\vzet}  \!= \!
  \setlength\arraycolsep{1.8pt}
  \begin{pmatrix} 
    \begin{bmatrix} 
    \setlength\arraycolsep{2pt}
    1 & \bet_1 &  \dots & \bet_1^n\\
    1 & \bet_2 &  \dots & \bet_2^n\\
    \vdots &  & &  \vdots \\
    1 & \bet_m & \dots & \bet_m^n\\
  \end{bmatrix}
  & \begin{bmatrix} 
  \setlength\arraycolsep{2pt}
    \bet_1^{n+1} & \bet_1^{n+2}  & \dots & \bet_1^{n\!+\!m\!-\!1}\\
    \bet_2^{n+1} & \bet_2^{n+2}  & \dots & \bet_2^{n\!+\!m\!-\!1}\\
     \vdots & \vdots & &   \vdots \\
     \bet_m^{n+1} & \bet_m^{n+2}  & \dots & \bet_m^{n\!+\!m\!-\!1}\\
  \end{bmatrix}\\
  \begin{bmatrix} 
  1 & \alp_1 &  \dots & \alp_1^n\\
  1 & \alp_2 &  \dots & \alp_2^n\\
  \vdots &  &   \vdots \\
  1 & \alp_n & \dots & \alp_n^n\\
  \end{bmatrix}
  & \begin{bmatrix} 
  \alp_1^{n+1} & \alp_1^{n+2}  & \dots & \alp_1^{n\!+\!m\!-\!1}\\
  \alp_2^{n+1} & \alp_2^{n+2}  & \dots & \alp_2^{n\!+\!m\!-\!1}\\
  \vdots & \vdots & &   \vdots \\
  \alp_n^{n+1} & \alp_n^{n+2}  & \dots & \alp_n^{n\!+\!m\!-\!1}\\
  \end{bmatrix}\\
 \end{pmatrix}
\end{align}
generated by the node vector $\vzet$. 
Hence, setting $\vlam=\vzet$ we get from the  observation \eqref{eq:lamnotation}
\begin{align}
  \vV_{\!\vzet}\vS_{\va\vb}\!=\! \vD \vV_{\!\valp,\vbet}
 \! =\!\!
 \setlength\arraycolsep{1pt}
 \begin{pmatrix} \vD_{\ua(\vbet)} & \vzero \\ \vzero &\! \vD_{\ub(\valp)}\end{pmatrix}
  \!\begin{pmatrix} 
    \vV_{\!\vbet} & \vzero\\
    \vzero &\vV_{\!\valp} 
  \end{pmatrix}\label{eq:sylvestervandermonde}
\end{align}
where $\vD=\vD_{\ua(\vbet),\ub(\valp)}$ is a diagonal matrix generated by evaluating the polynomials at the other
zeros and $\vV_{\valp}$ and $\vV_{\vbet}$ are $n\times n$ resp. $m\times m$ Vandermonde matrices
\begin{align}
\setlength\arraycolsep{1pt}
  \vD_{\ub(\valp)}\!=\!\begin{pmatrix}
    \ub(\alp_1) & 0 & \dots & 0\\
    0 & \ub(\alp_2) & \dots & 0\\
    \vdots &  & \ddots & \vdots\\
    0 & 0 & \dots & \ub(\alp_n)
  \end{pmatrix},\quad
  \vV_{\valp}\!=\!
  \setlength\arraycolsep{1pt}
  \begin{pmatrix} 
    1 & \alp_1 & \dots & \alp_1^n\\
    1 & \alp_2 & \dots & \alp_2^n\\
    \vdots & \vdots &  & \vdots\\
    1 & \alp_n & \dots & \alp_n^n\\
  \end{pmatrix}.\notag
\end{align}
Note, $\vD\vV_{\valp,\vbet}$ is not Hermitian. Taking the determinant of both sides in \eqref{eq:sylvestervandermonde} gives
\begin{align}
  \begin{split}
&\quad\quad\quad  \det(\vV_{\vzet}\vS_{\va,\vb}) = \det(\vD \vV_{\valp,\vbet})\\
&\LRA \quad \det(\vV_{\vzet}) \det(\vS_{\va,\vb})= \det (\vD)\det(\vV_{\valp})\det (\vV_{\vbet})
\end{split}
\end{align}
By using the well-known determinant formula for Vandermonde matrices, see \cite[Ch. IX, Exc.6]{LB99}, we have the
equivalence
\begin{align}
  & \det (\vS_{\va,\vb\!})\!\!\!\!\Pro_{1\leq l<k}^{n+m}\! \!\!(\zet_k\!-\!\zet_l)  
  \!=\!\!\Pro_{l,k=1}^{m,n} \!\!\!\ua(\bet_l) \ub(\alp_k)\! \!\Pro_{k<k'}^n\!\! (\alp_k\!-\!\alp_{k'})\!\Pro_{l<l'}^m
  \!(\bet_{l}\!-\!\bet_{l'})\notag\\
  &\LRA \quad  \det(\vS_{\va,\vb})\Pro_{1\leq l<k}^{n+m} (\zet_k-\zet_l)  
  =a_{0}^m b_0^n \Pro_{k,l=1}^{m,n} (\bet_l\!-\!\alp_k) \notag\\
  &\quad\quad\quad\quad\cdot \underbrace{\Pro_{l,k=1}^{n,m} (\alp_k-\bet_l) 
  \Pro_{k<k'}^n (\alp_k-\alp_{k'})\Pro_{l<l'}^m(\bet_{l}\!-\!\bet_{l'})}_{=\Del}\notag \\
  &\LRA \quad  \det(\vS_{\va,\vb}) \cdot \Del
  = a_{0}^m b_0^n\Pro_{l,k=1}^{m,n} (\bet_l-\alp_k) \cdot \Del
  \end{align}
Where $\Del\not=0$ if and only if all zeros are simple and $\ua$ and $\ub$ are coprime. If this holds, we can resolve for
the determinant 
\begin{align}
  \det(\vS_{\va,\vb})= a_0^m b_0^n \Pro_{k,l=1}^{m,n}(\bet_k-\alp_l)
\end{align}
which proofs \eqref{eq:detsab} if all zeros are simple.  On the other hand we can compute directly from the relation
\eqref{eq:sylvestervandermonde} for any $N\times N$ matrices
\begin{align}
  \begin{split}
  \sNorm{\vV_{\vzet}}_{\infty} \sig_{1}(\vS_{\va,\vb}) &\geq\sig_{1} (\vV_{\vzet}\vS_{\va,\vb})
  = \sig_1 (\vD \vV_{\valp,\vbet})\\
  &\geq \sig_1(\vV_{\valp,\vbet})\sig_{1}(\vD)
\end{split}
\end{align}
where the first inequality holds by%
\footnote{Note, in Bahtia, the singular values are in decreasing order, i.e., $\sig_1\geq \sig_2\geq \dots\geq
  \sig_r\geq \sig_{r+1}=0=\dots =\sig_n$.}
  \cite[Probl.III.6.2]{Bha96} and the last  by \cite[Thm.9]{MK04b} if $\vD$ and $\vV_{\valp,\vbet}$ are booth
non-singular.  Since both are block diagonal matrices we get further
\begin{align}
 \sig_1(\vV_{\!\valp,\vbet})\sig_{1}(\vD) 
 \!\geq\!\min\{\sig_{1}(\vV_{\!\valp}),\sig_1(\vV_{\!\vbet})\}
\cdot \del^{N\!-\!1}|a_0b_0|
\end{align}
where we used $\min\{|a_0|,|b_0|\}\geq |a_0|\cdot|b_0|$
since $|a_0|,|b_0|\leq 1$.
For the lower bound of the smallest and
largest singular values of Vandermonde matrices we can use bounds derived and referenced in \cite{AB17}.
Hence we get
\begin{align}
  \sig_1(\vS_{\va,\vb})&\geq \del^{N\!-\!1}|a_0b_0|\cdot \sNorm{\vV_{\vzet}}_{\infty}^{-1}\cdot
  \min\{\sig_{1}(\vV_{\valp}),\sig_1(\vV_{\vbet})\}\notag\\
  &\geq \frac{\del^{N\!-\!1}|a_0b_0|}{N}\cdot\Gam_1\cdot\Gam_2
\end{align}
where
\begin{align}
  \Gam_1\!&=\!\Big( \max\{\sum_{k=1}^N |\zet_k|^{N-1}, N\}\Big)^{-1}\\
  \Gam_2\!&=\! \min 
  \left\{\max_k \!\!\Pro_{k\not=k'=1}^n \!\!\!\frac{\max\{1,|\alp_{k'}|\}}{|\alp_k-\alp_{k'}|}
  ,\max_l\!\!\Pro_{l\not=l'=1}^m\!\!\!
  \frac{\max\{1,|\bet_{l'}|\}}{|\bet_{l'}-\bet_l|}\right\}\notag
\end{align}
We need lower and upper bounds for the smallest and largest zero in magnitude. Since the coefficient vectors are
normalized we get
\begin{align}
  \max_{k\not=l} |\alp_k -\alp_l|\leq 2 (1+\frac{\Norm{\va}}{|a_0|})\leq 2\Big(\frac{1+|a_0|}{|a_0|}\Big)
\end{align}
see for example \cite[Thm.2]{KS91}, and similar for $\ub$.
Then we get
\begin{align}
  \Gam_1 &\geq \frac{1}{N-1}\frac{1}{a^{1-N} + b^{1-N}} \\ 
  \Gam_2 &\geq \min\Big\{ (a/2)^{n-1} ,(b/2)^{m-1}\Big\}
  \intertext{where}
  a&=\frac{|a_0|}{|a_0|+1}\quad,\quad b=\frac{|b_0|}{|b_0|+1}
\end{align}
are less than $1/2$.
Hence we get
\begin{flalign*}
&&  \sig_1(\vS_{\va,\vb})&\geq \frac{\del^{N\!-\!1}|a_0b_0|}{N^2}\cdot
  \frac{\min\{(a/2)^{n-1},(b/2)^{m-1}\}}{a^{1-N}+b^{1-N}} && \qed
\end{flalign*}
\begin{remark}
  The singular value bounds of the Vandermonde matrices used in this proof are very weak since they include all worst
  cases. If there is more structure known about the zeros of $\ux_1$ and $\ux_2$ much tighter bounds are available, for
  example, if all nodes of $\valp$ are lying uniform on the unit circle. Such insights also may lead to signal designs
  where deconvolution is guaranteed to be stable or unstable.
\end{remark}

\section{Dual Certificate}\label{sec:dualcertificate}
To show that $\vW=\vS_{\vx_2^0,-\vx_1^0}\vS_{\vx_2^0,-\vx_1^0}$ defines a dual certificate we need to show the following
three properties
\begin{enumerate}[(i)]
  \item $\vW\mgeq 0$
  \item $\vW\vx=0$
  \item $\rank(\vW)=N-1$
\end{enumerate}
see also \cite{Jag16,JH16,WJPH17}.  By definition $\vW$ is already positive semi-definite and $\vW\vx=0$ by observing
the commutative property of the convolution $\vx_2*\vx_1-\vx_1*\vx_2=\zero$ in \eqref{eq:convdiff}. The last property
follows finally by the Sylvester Theorem due to the co-prime condition of $\ux_1$ and $\ux_2$ in
\eqref{eq:sylvesterrank}.  The missing injectivity property for the uniqueness result of \thmref{thm:4correlation}
follows then from the local stability result \lemref{lem:localstab}.  

% Figure for Toeplitz and Hankelmatrix (convolutoin and correlation)
%
\begin{figure*}[!t]
  \begin{subfigure}[b]{0.99\textwidth}
  \centering
     \includegraphics[width=1\textwidth]{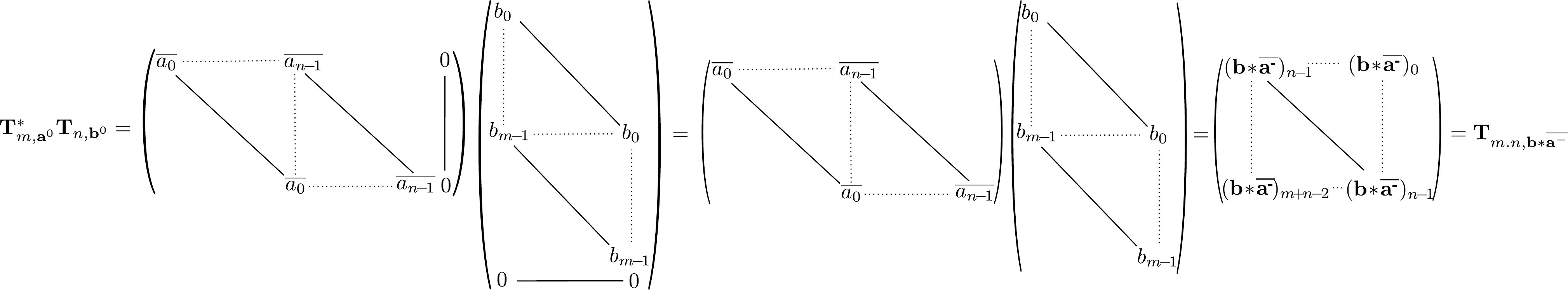}
    % \caption{Correlation via Toeplitz matrices.}
     \label{fig:ToeplitzCorrelation}
     \vspace{-1em}
  \end{subfigure}
  \caption{Correlation between the vectors $\va\in\C^{n}$ and $\vb\in\C^{m}$.}
  \label{fig:Toeplitzoperation}
\end{figure*}
The adjoint operator  of $\Alin$ in \eqref{eq:4maskmeasurements} is the linear map
\begin{align}
  \Alin^*\colon \C^M \to \C^{N\times N} 
\end{align}
which is defined for every 
$\vlam\in\C^M, \vX\in\C^{N\times N}$ as
\begin{align}
  \skprod{\Alin^*(\vlam)}{\vX}&=\skprod{\vlam}{\Alin(\vX)}=\skprod{\vlam}{\ve_m \tr(\vA_m\vX^*)}\label{eq:Alinadjoint}\\
  &=\skprod{\vlam}{\ve_m\skprod{\vX}{\vA_m}}=\skprod{\sum_m \lam_m \vA_m}{\vX}\notag
\end{align}
where the matrices $\vA_m$ are given by $\vA_{i,j,k}$ in \eqref{eq:3Nb} in lexicographical order. 
Since this holds for all $\vX\in\C^{N\times N}$ we get
\begin{align}
  \range(\Alin^*):=\spann\{\vA_0,\dots,\vA_{M-1}\}.
\end{align}
To show that the Sylvester matrix product $\vW$ in \eqref{eq:dualcertif} corresponds to the dual certificate of our
measurements, we need hence to find the corresponding vector $\vlam$. To see this, we have to represent $\vW$ in
terms of the measurements.
If we consider the product of the adjoint Sylvester matrix with itself we get by \eqref{eq:dualcertif} and 
\eqref{eq:sylvesterx}
\begin{align}
\vS_{\vx_2^0,-\vx_1^0}^* \vS_{\vx_2^0,-\vx_1^0} 
&=
  \setlength\arraycolsep{3pt}
  \begin{pmatrix}   \vT_{1,\vx_2}^* & \zero \\ -\vT_{2,\vx_1}^* & \zero \end{pmatrix}
  \begin{pmatrix}   \vT_{1,\vx_2} & -\vT_{2,\vx_1} \\ \zero^T & \zero^T\end{pmatrix} \notag \\
&=
\setlength\arraycolsep{3pt}
\begin{pmatrix}  
  \vT_{1,\vx_2}^*\vT_{1,\vx_2}  & -\vT_{1,\vx_2}^*\vT_{2,\vx_1}\\
  -\vT_{2,\vx_1}^*\vT_{1,\vx_2} & \vT_{2,\vx_1}^*\vT_{2,\vx_1}\\
\end{pmatrix}\label{eq:sylvesterrangestar}
\end{align}
see \figref{fig:Toeplitzoperation}. Let $i,j,i',j'\in\{1,2\}$ with $i+j=i'+j'$,
then we set $\Nij=\Ni+\Nj=\Nip+\Njp$ (we have $N_{1,2}=N$), we get with \eqref{eq:Tform} for the banded Toeplitz products
and the embedding matrix notation in \eqref{eq:elementtoeplitz}

\begin{align}
  \vT_{j,\vx_i^0}^*\vT_{j',\vx_{i'}^0}\!&=\vT_{\Nj,\vx_i^0}^*\vT_{\Njp,\vx_{i'}^0} \\
&= \sum_{m=0}^{\Ni-1} \sum_{l=0}^{\Nip-1}\cc{x_{i,m}} x_{i',l}\Proj_{i+j,j}^T \vT_{\Nij}^{-m}
\vT_{\Nij}^{l}\Proj_{i+j,j}\notag\\
    &=\! \sum_{m=0}^{\Ni-1} \!\sum_{l=0}^{\Nip-1}\cc{x_{i,m}} x_{i',l}\Proj_{i+j,j}^T \vT^{l-m}_{\Nij}
    \Proj_{i+j,j'}.\notag
\end{align}
Let us emphasize that $l,m$ are limited by $\Ni\!-\!1$ resp. $\Nip\!-\!1$, and since we consider the $\Ni$ resp. $\Nip$
embeddings the zeros on the $l-m$th diagonal in $\vT^{l-m}_{\Nij}$ can be ignored. By substituting with $k=l-m$ we get
\begin{align}
  \!\!\vT_{\!j,\vx_i^0}^*\!\vT_{\!j'\!,\vx_{i'}^0}\!\!  &=\!\!\sum_{k=-\Ni+1}^{\Nip-1}\!\!\!\Big(\sum_{m=0}^{\Ni-1} \cc{x_{i,m}}
  x_{i',k+m}\!\Big) \Proj_{i+j,j}^T\vT_{\Nij}^{k}\Proj_{i+j,j'}\notag 
\intertext{where we set $x_{i,k}=0$ for $k<0$ or $k\geq N_i$. Bringing $m$ on the other side by substituting $m'=m+k$ gets} 
  &=\!\!\sum_{k=1\!-\!\Ni}^{\Nip-1}\!\!\!\Big(\!\sum_{m'=k}^{\Ni-1+k}\!\! \!\cc{x_{i,m'-k}} x_{i',m'}\!\Big)
  \Proj_{i+j,j}^T\vT^{k}_{\Nij}\Proj_{i+j,j'}\notag\\
  &=\!\!\sum_{k=-\Ni\!+\!1}^{\Nip-1}\!\! \!\!(\vx_{i'}\!*\!\vxict)_{\Ni\!-\!1\!+\!k}
  \Proj_{i\!+\!j,j}^T\vT^{k}_{\!\Nij}\Proj_{i\!+\!j,j'} 
\intertext{where the inner  sum is the correlation between $\vx_{i'}$ and $\vx_{i}$. Using the matrix notation
$\vT_{\Nj,\Njp}^{(k)}$ defined in \eqref{eq:LiLj} we get}
  &=\!\sum_{k=0}^{N_{ii'}-2}\! (\vx_{i'} * \vxict)_k \vT_{\Nj,\Njp}^{(k)}  
\end{align}
such that we have
\begin{align}
  \vT_{j,\vx_i^0}^*\!\vT_{j',\vx_{i'}^0}\!=\!
  \vT_{j,\vx_i}^*\!\vT_{j',\vx_{i'}}\!=\!\vT_{j,j',\vx_{i'}*\svxict}\label{eq:convTopelitzproduct}.
\end{align}
Hence, the time-reversal and complex conjugation for the correlation is obtained  by the adjoint operation. Hence we get 
\begin{align}
 \vS_{\vx_2^0,-\vx_1^0}^* \vS_{\vx_2^0,-\vx_1^0} 
 \setlength\arraycolsep{3pt}
 &=\begin{pmatrix} 
   \vT_{1,1,\va_2} & -\vT_{1,2,\va_{1,2}}\\
  - \vT_{2,1,\va_{2,1}} & \vT_{2,2,\va_1} 
\end{pmatrix}\label{eq:dualcertificateTblocks}
\end{align}
Note, that $\vT_{2,1,\va_{2,1}}=\vT_{2,1,\svaonetwoct}$.  Note, that $\vome\not=\va$ since we need a minus for the
anti-diagonal and $\va_1$ and $\va_2$ are interchanged.  More precisely we have
\begin{align}
  \vome:=[\va_2^{(1)}, -\va_{1,2},-\va_{2,1},\va_{1}^{(2)}]
\end{align}
where $\va_i^{(j)}$ are zero padded if $\Nj>\Ni$ and truncated otherwise, such that the dual certificate is given by
\begin{align}
  \vW&=\vS^*\vS= \Alin^*(\vome)\notag\\
  &=\!\!\sum_{k=\Ltwo-N_1}^{2\Lone-2} a_{2,k} \vA_{1,1,k} + \sum_{k=0}^{2\Ltwo-2} a_{2,k}\vA_{2,2,k} \\
 &\quad -\sum_{k=0}^{N-2}\Big( a_{1,2,k}\vA_{1,2,k} +a_{2,1,k}\vA_{2,1,k}\Big)\notag
 \label{eq:AlamW}
\end{align}
where $a_{i,k}=0$ for $k<0$ and $k>2\Ni-2$ and the sensing matrices $\vA_{i,j,k}$ in \eqref{eq:4maskmeasurements}. \\

   % Trace norm type inequalities
\renewcommand{\ve}{\boldsymbol{e}}
\section{Trace Norm Inequalities}

\begin{lemi}\label{lem:tracelowerbound}
Let $\vX\mgeq 0$ and $\vY\succ 0$. Then it holds
\begin{align}
  \tr(\vX\vY)\geq \lam_{1}(\vY)\sNorm{\vX}_1
\end{align}
where $\lam_1(\vY)$ is the smallest eigenvalue of $\vY$.
\end{lemi}
\begin{proof}
Observe, that by  definition of a positive semi-definite matrix $\vA$, i.e., 
\begin{align}
  \skprod{\vx}{\vA\vx}\geq 0 \quad,\quad \vx\in\K^N
\end{align}
for any field $\K\in\{\C,\R\}$, we get with $\vx=\ve_k$ 
\begin{align}
  \skprod{\ve_k}{\vA\ve_k} = a_{k,k}\geq 0 \quad,\quad k\in[N],
\end{align}
i.e., the entries are non-negative if the matrix is positive semi-definite and positive if the matrix is positive. 

Hence, the trace of a positive semi-definite matrix is non-negative, and we can lower bound the trace by a non-negative
constant.  We diagonalize $\vY$ by an unitary matrix $\vU$ such that
\begin{align}
  \tr(\vX\vU\vU^*\vY\vU\vU^*) = \tr(\vU^*\vX\vU\vD_{\vlam})
\end{align}
where $\vlam=\vlam(\vY)$ are the eigenvalues written as a vector in increasing order and $\vD_{\vlam}$ is a diagonal
matrix generated by the vector $\vlam$. We can hence write the trace as
\begin{align}
  \sum_{k=1}^N\lam_k(\vY) \tr(\vX\ve_k\ve_k^*)
   &=\sum_k \lam_k(\vY)\underbrace{x_{k,k}}_{\geq 0}\\
   &\geq \lam_{1}(\vY))\tr(\vX)= \lam_{1}(\vY)\sNorm{\vX}_1\notag
\end{align}
\end{proof}

\section{Some Remarks for the Local Stability Proof}
In \eqref{eq:dwowo} we can lower bound  
\begin{align}
  \Norm{\vx_1}_2^4+\Norm{\vx_2}_2^4
\end{align}
by using $\nu=\Norm{\vx_1}_2^2$ and
$(1-\nu)=\Norm{\vx_2}_2^2$ with $\nu\in[0,1)$, since  $\Norm{\vx_1}_2^2 + \Norm{\vx_2}_2^2= \Norm{\vx}_2^2=1$. Hence the function
\begin{align}
  f(\nu)=2\nu^2 -2\nu +1
\end{align}
is minimized for $\nu\in[0,1)$ if the derivative vanishes, i.e., 
\begin{align}
   0=4\nu  -2 \quad \RA \quad \nu=1/2.
\end{align}
Hence we get for the global minimum
\begin{align}
  \min_{\Norm{\vx_1}_2^2+\Norm{\vx_2}_2^2=1}\Norm{\vx_1}_2^4+\Norm{\vx_2}_2^4 = \frac{1}{2}\label{eq:lagrange}
\end{align}
which is achieved for $\Norm{\vx_1}_2=\Norm{\vx_2}_2$.

\end{document}